\documentclass[12pt,a4paper]{article}
\usepackage{float}
\usepackage[english]{babel}
\usepackage{latexsym}
\usepackage {amssymb}
\usepackage {amsmath}
\usepackage {mathtools}
\usepackage{amsthm}
\usepackage{color}
\usepackage{hyperref}
\usepackage{setspace}
\usepackage{booktabs}
\usepackage{multirow}
\usepackage[titletoc,title]{appendix}

\usepackage{rotating}

\usepackage[longnamesfirst]{natbib}
\bibpunct{(}{)}{;}{a}{,}{,}
\renewcommand{\cite}{\citep}

\def\beqn{\begin{eqnarray*}}
\def\eeqn{\end{eqnarray*}}
\def\beq{\begin{eqnarray}}
\def\eeq{\end{eqnarray}}

\newtheorem{theorem}{Theorem}

\newtheorem{remark}{Remark}

\makeatother
\makeatletter
\renewcommand{\@fnsymbol}[1]{\@arabic{#1}}
\makeatother

\usepackage{graphicx}
\usepackage{mathrsfs}
\graphicspath{{Figures/}}

\title{Return-to-Baseline Testing \\via Empirically Calibrated e-processes}

\author{
Marta Regis
\footnote{
Department of Mathematics and Computer Science,
Eindhoven University of Technology,
PO Box 513, 5600 MB Eindhoven,
the Netherlands;
\href{mailto:m.regis@tue.nl}{m.regis@tue.nl}.}
\and
Paulo Serra
\footnote{
Department of Mathematics,
Vrije Universiteit Amsterdam,
De Boelelaan 1105, 1081 HV Amsterdam,
the Netherlands;
\href{mailto:p.j.de.andradeserra@vu.nl}{p.j.de.andradeserra@vu.nl}.
Corresponding author.}
}

\begin{document}

\baselineskip=25pt

\maketitle

\begin{abstract}
\baselineskip=15pt
\noindent
We consider the problem of detecting a Return to Baseline (RtB) in high-frequency monitoring data preceding and following an intervention, where the aim is to identify the time at which the data-generating distribution realigns with its pre-intervention distribution. We propose a sequential, distribution-free testing procedure that does not rely on specifying a null model and provides anytime-valid error control. The method relies on ideas from universal inference to define a discrepancy measure that is aggregated into a non-negative super-martingale, and is then empirically calibrated to form an e-process. The calibration is performed using the baseline data, and is thus subject-specific. We establish finite-sample bounds for the calibration error (under a flexible non-parametric assumption), discuss the impact of tuning parameters and computational complexity, and illustrate through simulations and a clinical case study that the procedure accurately detects RtB from monitoring data.\\

\noindent {\textit{Key words and phrases.}}
anytime-valid inference,
e-process,
high-frequency monitoring,
return-to-baseline,
sequential testing,
universal inference.
\end{abstract}

\newpage

\section{Introduction}
\label{section:introduction}

In many monitoring problems an intervention induces a temporary departure from a subject-specific baseline, after which the process gradually returns to its pre-intervention behaviour. 
At the (unknown) time at which this occurs, we say that a \emph{return to baseline} (RtB) has occurred. 
A common example comes from clinical monitoring, where a clinician administers an intervention whose effect is known to decay over time, while physiological variables are observed continuously or at high frequency. 
From a statistical perspective, the goal is to determine when the data-generating distribution has realigned with its pre-intervention distribution, based solely on the observed data.

Detecting an RtB is a challenging sequential problem. 
Monitoring may proceed continuously, but may also be paused or resumed, and decisions are made based on data observed up to data-dependent times. 
Moreover, the baseline behaviour is typically complex and subject-specific, and the effect of the intervention is often unknown beyond the fact that it perturbs the distribution of the process. 
These features distinguish the RtB problem from classical hypothesis testing like Wald’s sequential likelihood ratio test~\citep{wald1947sequential} which provides optimal procedures for testing simple hypotheses under fully specified parametric models, or cumulative sum (CUSUM) schemes~\citep{page1954continuous} where inference is framed in terms of detecting deviations from a null model, and from statistical process control, which is primarily concerned with identifying departures from a target state~\citep{shewhart1930economic,montgomery2020introduction}. 
Sequential two-sample testing~\citep{lheritier2015sequential, shekhar2023nonparametric, podkopaev2023sequential}
provides another related (but distinct) framework.
There, one observes data from two streams simultaneously and tests whether their
underlying distributions coincide, building e-values or martingales that grow whenever
the null of equality is violated.
The RtB problem differs in two fundamental respects: the baseline distribution $P$ is
unknown and must be estimated from a sample, making the null hypothesis
composite; and the inferential goal is not to detect a \emph{departure} from
distributional equality but to certify a \emph{return} to a previously observed
distribution.
In the RtB setting, baseline data are available and play a central role: rather than asking whether the process deviates from baseline, the task is to certify when it can again be regarded as consistent with the baseline distribution.

Despite its practical relevance, RtB is rarely treated as a statistical estimation or testing problem based on the observed measurement process. 
In clinical studies, RtB is frequently recorded as an observed or reported outcome -- often a binary indicator collected via questionnaires -- and used to compare treatments, patient groups, or time-to-event outcomes~\citep{dahl2009prospective,bergeron2022return,steuart2020discharge,jayadevappa2007ethnic,martin2007health,rossi2016erectile}. 
The RtB is thus not inferred from the data stream itself and one commonly relies on prespecified thresholds for what constitutes a meaningful change~\citep{jayadevappa2007ethnic}.

Current monitoring technologies increasingly produce dense and continuous data streams, and this creates opportunities for more precise, data-driven inference. 
For instance, in fetal and maternal monitoring studies investigating the effects of antenatal corticosteroid administration, where repeated or continuous measurements of heart rate variability and photoplethysmography signals are available~\citep{de2008differential,noben2019fetal,bester2023changes,bester2024changes}. 
Early analyses relied on repeated-measures ANOVA and pairwise contrasts, implicitly treating the intervention window as known. 
With richer data, it becomes natural to infer RtB directly from the observed signal, without specifying a parametric model for either baseline or intervention effects.

Any statistical procedure for detecting an RtB in such settings should satisfy several requirements. 
It should be calibrated to the individual baseline rather than to a population-level (parametric) model, and thus aim to be largely distribution-free. 
It should remain valid under continuous monitoring and optional stopping, so that sequential decision-making does not invalidate error guarantees. 
It should be robust to outliers and model misspecification, as monitoring data are often noisy and heterogeneous, and to interruptions in the data stream as these are also common.
Finally, it should avoid imposing restrictive assumptions on the intervention effect, which is typically unknown and likely to vary over time and patient.
In this paper we develop a sequential RtB testing procedure that meets exactly these requirements. 

Our approach is based on \emph{e-processes}~\citep{shafer2011test,grunwald2020safe,ramdas2020admissible,ramdas2024hypothesis}, which provide a natural framework for anytime-valid inference via non-negative super-martingales. 
While many existing e-process constructions rely on parametric likelihoods or explicitly specified null models, we take a different approach: the e-process is calibrated empirically using baseline data alone, without explicitly specifying a null sub-model.
Conceptually, our method draws on ideas from \emph{universal inference}~\citep{wasserman2020universal} and its recent implementations~\citep{dey2025generalized}, but differs from that literature by avoiding explicit model specification and by focusing on certifying a return to a previously observed distribution rather than testing a fixed null hypothesis.

The main contributions of this paper are the following.
We formulate the return-to-baseline problem as a sequential testing task for high-frequency monitoring data, explicitly accounting for continuous monitoring and optional stopping.
We propose a distribution-free RtB detection procedure based on an empirically calibrated e-process, where calibration relies solely on baseline data and is therefore subject-specific.
We establish finite-sample bounds on the calibration error under mild, flexible nonparametric assumptions, clarifying the role of tuning parameters and their effect on conservativeness and power.
We analyse the computational complexity of the proposed method and discuss practical trade-offs relevant for high-frequency data.
Finally, we demonstrate the performance of the procedure through simulations and a clinical case study, showing accurate RtB detection in a realistic monitoring scenario.

The remainder of the paper is organised as follows. Section~\ref{section:framework} introduces the formal RtB framework and the construction of the empirically calibrated e-process.
Section~\ref{section:numerics} illustrates the procedure, and Section~\ref{section:case_study} presents a case study.
We close the article with some conclusions in Section~\ref{section:conclusions}.
The proof of our main result in Appendix~\ref{section:proof}, and
further numerical results related to the computational performance of the algorithm are in Appendix~\ref{section:numerics:computational_cost}.

\section{The Return-to-Baseline framework}
\label{section:framework}

\subsection{The generic RtB problem}
\label{section:framework:rtb}

Before we specify a model for the data, the general RtB problem consists of the following.
We observe $n$ random vectors
\[
X_{1}, \dots, X_{n}\stackrel{i.i.d.}{\sim} P,
\]
where $P$ is the measure underlying the unknown joint distribution of the data under the baseline.
After these $n$ observations (with $n$ known), an intervention takes place that affects the distribution of the following $r$ measurements so that
$$
X_{n+i} \sim P_i, 
\quad
i=1,\dots, r.
$$
Here, the $P_i$'s are the (unknown, different from $P$, and potentially mutually different) measures underlying the joint distributions under the intervention.
After these $r$ observations, an RtB has occurred and data is again sampled from $P$ so that from that point onward we observe
\[
X_{n+r+1}, X_{n+r+2}, \dots \stackrel{i.i.d.}{\sim} P.
\]

The RtB problem consists of detecting $r$ when sequentially observing
\[
X_{1}, \dots, X_{n},X_{n+1}, \dots, X_{n+r},
X_{n+r+1}, \dots, X_{n+r+m}, \dots,
\]
where $n$ is known, while keeping $P$ as unspecified as possible.

\subsection{A statistical model for quasi-periodic data}
\label{section:framework:data_model}

Motivated by the case study that we discuss in Section~\ref{section:case_study}, we focus on the following model for the data.
Consider a stochastic process $\mathscr{X}=(\mathscr{X}(t))_{t\ge0}$ and, for any $\tau\in(0,1)$, its respective conditional quantiles $f^{(\tau)}$ defined as satisfying (by definition)
\begin{equation}
\label{equation:quantiles}
P\Big(\mathscr{X}(t)\le f^{(\tau)}(t) \Big) = \tau,
\quad t\ge0.
\end{equation}
The process $\mathscr{X}$ has some underlying periodicity in that for some $\rho>0$,
\begin{equation}
\label{equation:model}
f^{(\tau)}(t + 1/\rho) = 
f^{(\tau)}(t), 
\quad \tau\in(0,1), t\ge0,
\end{equation}
We refer to such a stochastic process (and data collected from it) as being \emph{quasi-periodic}.
We refer to its period as being $1/\rho$, or its frequency as being $\rho$.

Suppose that $t_{i,j} = (i-1)/\rho + t_j$, for $i = 1, \dots, n$, where $t_j = (j-1)/\rho$, for $j=1, \dots, p$,
so that, in particular, $f^{(\tau)}(t_{i,j}) = f^{(\tau)}(t_{j})$.
We assume that the data are collected from the stochastic process $\mathscr{X}$ at times $t_{i,j}$, and denote
\begin{equation}
\label{equation:observations}
X_{i,j} = 
\mathscr{X}(t_{i,j}), \quad
i = 1, \dots, n,\;
j = 1, \dots, p.
\end{equation}
We abbreviate the equally-spaced observations collected during the $i$-th period as $X_i=(X_{i,1}, \dots, X_{i,p})$.

Under~\eqref{equation:quantiles}-\eqref{equation:observations}, the respective data vectors $X_i$ are identically distributed, and these conditions specify our model for the data \emph{under the baseline}.
(See Remark~\ref{remark:independence} about the role of independence.)
Under the intervention, we still assume that observations are taken according to~\eqref{equation:observations} but that the intervention (temporarily) affects the distribution of $\mathscr{X}$.
In fact, the random vectors measured during the intervention are only assumed to have a different distribution from that of the ones taken during the baseline period, but are otherwise allowed to have a completely arbitrary distribution.
In particular, it is not assumed that the data remain quasi-periodic under the intervention.

\subsection{Example of data models}
\label{section:framework:examples}

We present a few increasingly more general examples of data models that fall into the framework in~Section~\ref{section:framework:data_model} with pre-prescribed quantile functions.


\noindent\textbf{Example 1}
Suppose that for periodic functions $f,\sigma$, with $\sigma>0$,
\[
X_{i,j} = f(t_{i,j}) + \sigma(t_{i,j})\, \epsilon_{i,j}, \quad
i = 1, \dots, n, \;
j = 1, \dots, p,
\]
where $\epsilon_{i,j}$ form a random sample from some distribution with median 0, variance 1, and quantile function $Q$.
In this case we have
\[
f^{(\tau)}(t) =
f(t) + \sigma(t) Q(\tau),
\]
where $f^{(1/2)}=f$ is the median regression function.
Setting $\sigma(\cdot)$ to a constant and $Q$ to the quantile function $\Phi^{-1}$ of a standard Gaussian gives an homoskedastic Gaussian regression model.

\noindent\textbf{Example 2}
One can also easily generate data with pre-specified quantile functions of any form.
Consider a family of functions 
$\{f^{(\tau)}, \tau\in(0,1)\}$ so that for each $\tau\in(0,1)$, $f^{(\tau)}$ is periodic, and so that for each $t\ge 0$, $\tau\mapsto f^{(\tau)}$ is strictly monotone increasing.
Consider a (marginally) ${\rm Unif}(0,1)$ distributed random vector 
$\big(U_{i,j}, 
i = 1, \dots, n, j = 1, \dots, p\big)$,
and define
\[
X_{i,j} = f^{(U_{i,j})}(t_{i,j}), \quad
i = 1, \dots, n, \;
j = 1, \dots, p.
\]
We have by construction, as required, that
\[
P\Big(X_{i,j}\le f^{(\tau)}(t_{i,j}) \Big) =
P\Big(f^{(U_{i,j})}(t_{i,j})\le f^{(\tau)}(t_{i,j}) \Big) =
P\Big(U_{i,j}\le \tau \Big) = \tau,
\]
where we use the monotonicity of $\tau\mapsto f^{(\tau)}$.


\noindent\textbf{Example 3}
Besides specifying quantiles (and therefore marginal distributions) of the process $\mathscr{X}$, one may also want to specify a correlation structure.
One can do this by considering a copula process; cf.~\cite{wilson2010copula}.
Consider a Gaussian Process $\mathscr{Z}$ with mean zero and some covariance kernel $\psi$ so that $\mathscr{U}(t)=\Phi\big(Z(t)/\psi(t,t)\big)$ is a process with uniform marginals.
One can then set $\mathscr{X}(t) = f^{(\mathscr{U}(t))}(t)$.
The covariance structure of $\mathscr{X}$ can be written down only implicitly (with the help of a Gaussian copula) so it may not be so transparent how a covariance structure for $\mathscr{Z}$ translates into a covariance structure for $\mathscr{X}$.
However, one can easily check that Spearman's $\rho$ (not to be confused with the frequency $\rho$ of the data process) and Kendall's $\tau$ are respectively,
\[
\rho(s,t) =
\frac{6}{\pi}\arcsin\left(\frac{\psi(s,t)}{2\psi(s,s)\psi(t,t)}\right), \quad
\tau(s,t) =
\frac{2}{\pi}\arcsin\left(\frac{\psi(s,t)}{\psi(s,s)\psi(t,t)}\right).
\]
So one can then control the strength of the auto-covariance in $\mathscr{X}$ by choice of the covariance kernel $\psi$.
(Again we refer the reader to Remark~\ref{remark:independence} about the role of independence.)

\vspace{1em}

As these examples illustrate, the model presented above, both as a working model and as a data generating model, is extremely flexible.
We now specify a model for the quantile functions $f^{(\tau)}$.

\subsection{A model for the quantile functions}
\label{section:framework:quantile_model}

To model the quantile functions $f^{(\tau)}$, we consider a linear model parametrised by coefficients
$\theta^{(\tau)}=(\theta^{(\tau)}_1, \dots, \theta^{(\tau)}_{d})$, $d\in\mathbb{N}$.
Specifically, we model $f^{(\tau)}$ as
\begin{equation}
\label{equation:function_model}
f_{\theta^{(\tau)}}(t) = 
\sum_{i=1}^{d+m}\theta^{(\tau)}_i B_i^m(\rho\, t),
\quad 0\le t\le 1/\rho,
\end{equation}
where 
$\{B_i^m, i=1,\dots,d+m\}$ 
forms a B-spline basis of order $m\in\mathbb{N}$ on $[0,1]$, relative to the knots
$0=l_0 \le l_1 \le \cdots \le l_{d+1}=1$, and where
$\theta^{(\tau)}_{d+i}=\theta^{(\tau)}_i$, for $i=1,\dots,m$.
Our motivation for using this model for the quantile functions is that one can easily check that 
$f_{\theta^{(\tau_1)}} \le f_{\theta^{(\tau_2)}}$ (meaning $f_{\theta^{(\tau_1)}}(t) \le f_{\theta^{(\tau_2)}}(t)$ for all $t\ge0$) if, and only if,
$\theta^{(\tau_1)} \le \theta^{(\tau_2)}$ (meaning if $\theta^{(\tau_1)}_i \le \theta^{(\tau_2)}_i$ for all $i=1,\dots,k$), c.f.,~Section 4.6 of~\cite{schumaker2007spline}.
This makes it simple to check (or enforce) the required monotonicity of the quantile functions.

Let $\hat f^{(\tau,-k)}$ be an estimate of the quantile function of level $\tau$ from the baseline data $X_1,\dots, X_{k-1}$, $X_{k+1}, \dots, X_n$, i.e., when observation $X_k$ is held out ($k=1,\dots,n$).
It will become clear later why we may omit baseline periods.
Under model \eqref{equation:function_model}, 
$\hat f^{(\tau,-k)}$ is of the form
$f_{\hat\theta^{(\tau,-k)}}$ where $\hat\theta^{(\tau,-k)}$
is obtained as minimiser of
\[
\theta \mapsto
\sum_{\substack{i=1\\ i\neq k}}^{n}
\sum_{j=1}^p
h_\tau\big(
X_{i,j} - f_{\theta}(t_j)
\big),
\]
where 
$h_\tau(x)=(\tau-1\{x<0\})x$ is the check function~\citep{koenker2001quantile}.

The minimisation of this criterion can be reformulated as a linear program and solved efficiently: quantiles of different levels can be estimated concurrently by minimising the sum of the respective criteria, and crossings of the respective estimated quantiles can be avoided by adding extra linear constraints to the linear program~\citep{bondell2010noncrossing}.

\subsection{A metric for detecting deviations from the baseline}
\label{section:framework:deviation_metric}

To detect deviations from the baseline, we estimate (several) conditional quantiles $f^{(\tau)}$ from the baseline observations ($X_1,\dots, X_n$).
We then use these estimates as a template to judge how well subsequent periods ($X_{n+1},\dots$) align with the baseline.

For each $x\in\mathbb{R}^p$, each set 
$S\subseteq\{1,\dots,p\}$, each index 
$k\in\{1,\dots,n\}$, and each quantile level
$\tau\in(0,1)$,
we define the count function
\[
N^{(\tau,-k)}(x, S) =
\big|\{j\in S: 
x_j > \hat f^{(\tau,-k)}(t_j)\}\big|.
\]

If $\hat f^{(\tau,-k)}$ were the true quantile function then, under the baseline model,
$N^{(\tau,-k)}(X_i, S)$ should be approximately $(1-\tau)|S|$ for each $X_i$ from the baseline;
Theorem 2.2 of~\cite{koenker2005quantile} makes this statement precise.
As such it makes sense to use the statistics
\begin{equation}
\label{equation:discrepancy}
T^{(\tau,-k)}(X_i,S) = \big|N^{(\tau,-k)}(X_i, S)-(1-\tau)|S|\big|,
\end{equation}
to measure discrepancy between each $X_i$ and the baseline (at quantile level $\tau$).

Typical realizations of these  discrepancies satisfy the following:
\begin{enumerate}
\def\theenumi{\alph{enumi}}
\def\labelenumi{(\theenumi)}
\item 
$T^{(\tau,-k)}(X_i,S)$, $i=1,\dots,n$, $i\neq k$, is small by design, for any $S$, $\tau$, $k$;

\item
$T^{(\tau,-k)}(X_k,S)$
is representative of the discrepancy for an independently observed period coming from the baseline, for any $k$;

\item
$T^{(\tau,-k)}(X_{n+i},S)$, $i=1,\dots,r$ is, for each $k$, large for at least some combination of $\tau$ and $S$ if $P_i$ is different from $P$;

\item
$T^{(\tau,-k)}(X_{n+r+i},S)$, $i=1,\dots$ is distributed like $T^{(\tau,-k)}(X_k,S)$.
\end{enumerate}

We are now ready to define an e-process based on the statistics in~\eqref{equation:discrepancy}.
This e-process will track discrepancy from the baseline.
Large values of this process indicate that we are under the intervention, while small values certify a return to baseline.

\subsection{An e-process to detect the return to baseline}
\label{section:framework:E_process}

Pick $s\in\{1,\dots,p\}$, consider\footnote{
The parameter $s$ can be set freely. 
It affects the precision and computational burden of the algorithm in a way that we make precise in Theorem~\ref{theorem:bound} and Section~\ref{section:numerics:computational_cost}, respectively.
} 
$\mathcal{S}_s=\{S\subseteq\{1,\dots,p\}:|S|=s\}$, and let $Z\sim {\rm Unif}(\mathcal{S}_s)$ be a subset with cardinality $s$ of $\{1,\dots,p\}$, chosen uniformly at random.

Consider also the $\sigma$-algebra $\mathcal{F}_0 =
\sigma(X_1,\dots,X_{n})$ generated by the data collected under the baseline.
For any collection $\mathcal{T}=\{\tau_1,\dots,\tau_q\}$, 
$q\in\mathbb{N}$, and
$k\in\{1,\dots,n\}$,
\begin{equation}
\label{eq:empirical_distribution_max_discrepancy}
\max_{\tau\in\mathcal{T}}
T^{(\tau,-k)}(X_k,Z)\mid \mathcal{F}_0 \sim
{\rm Unif}\big(\{
\max_{\tau\in\mathcal{T}}T^{(\tau,-k)}(X_k,S): 
S \in \mathcal{S}_s\}\big).
\end{equation}
This means that we can empirically derive the distribution of the maximal discrepancy over quantile levels $\mathcal{T}$ for baseline period $k$,  for a randomly chosen set of indices $Z$ with cardinality $s$.
Since our estimate of the conditional quantiles does not depend on $X_k$, this distribution is representative of the distribution for the maximum discrepancy for a new sample from the baseline.
This is the key to calibrate the super-martingale that we define in the following.

\begin{remark}
\label{remark:independence}
One is likely to encounter auto-correlations in $\mathscr{X}$ so the reliance on independent observations across different periods in the generic RtB problem in Section~\ref{section:framework:rtb} may sound restrictive.
However, one can select the sets in $\mathcal{S}_s$ as subsets of cardinality $s$ of $\{r_1,\dots,r_2\}$ for $1<r_1<r_2<p$, $r_2-r_1> s \ge 1$, rather than as subsets of $\{1,\dots,p\}$.
In doing so, we create a buffer between the time intervals corresponding to the observation vectors, reducing possible correlation between them.
\end{remark}

Define, for an $\mathcal{F}_0$-measurable random variable $D_\gamma>1$ (a.s.) to be specified in Section~\ref{section:framework:E_process:norming}, 
the stochastic processes $\mathscr{M}^{(-k)}=(\mathscr{M}_t^{(-k)})_{t\in\mathbb{N}}$, $k=1,\dots,n$, with
\begin{equation}
\label{eq:dicrepancy_process}
\begin{aligned}
\mathscr{M}_1^{(-k)} &= 
\frac1{D_\gamma}
\max_{\tau\in\mathcal{T}}T^{(\tau,-k)}(X_{n+1},Z),\\
\mathscr{M}_t^{(-k)} &=
\min\left(\mathscr{M}_{t-1}^{(-k)}, \;
\frac1{D_\gamma}{\max_{\tau\in\mathcal{T}}T^{(\tau,-k)}(X_{n+t},Z)}\right),
\quad t = 2,\dots,
\end{aligned}
\end{equation}
so that for each $k$, $\mathscr{M}_t^{(-k)}$ is a running minimum of similar terms.

Consider additionally the filtration $\mathcal{F} =
(\mathcal{F}_t)_{t\in\mathbb{N}_0}$ 
\[
\mathcal{F}_t =
\sigma(X_1, \dots, X_{n+t}, Z),
\quad t\in\mathbb{N},
\]
where we remind that $Z\sim {\rm Unif}(\mathcal{S}_s)$, independent of the data.
Note that, by definition, each process $\mathscr{M}^{(-k)}$ is non-negative, adapted to the filtration $\mathcal{F}$, and non-increasing
so that
\[
E[\mathscr{M}_t^{(-k)}\mid\mathcal{F}_{t-1}] \le
E[\mathscr{M}_{t-1}^{(-k)}\mid\mathcal{F}_{t-1}] = \mathscr{M}_{t-1}^{(-k)}.
\]
Since the maximum in each $\mathscr{M}_1^{(-k)}$ is bounded by $s$, we conclude that the processes $\mathscr{M}^{(-k)}$ are non-negative super-martingale with respect to $\mathcal{F}$, and then e-processes by appropriate choice of $D_\gamma$.
Note that the same is then also automatically true for the process
$\mathscr{M}=(M_t)_{t\in\mathbb{N}}$,
where
\[
\mathscr{M}_t = 
\frac{1}{n} \sum_{k=1}^n \mathscr{M}_t^{(-k)}.
\]

\begin{remark}
In our context, considering a running minimum rather than a running average- or product makes the most sense.
The super-martingale scales with discrepancy from the baseline and should remain large as long as we are under the intervention.
The moment that data coming from the null gets fed to the discrepancy metric $T$, this results in small discrepancies that causes the super-martingale to drop.
This allows the RtB to be detected without lag -- a running average or product would likely take several time units to drop sufficiently to indicate an RtB. 
\end{remark}

\subsubsection{Empirical calibration of the e-process}
\label{section:framework:E_process:calibration}

The random variable $D_\gamma$ in~\eqref{eq:dicrepancy_process} is chosen so as to guarantee that the expectation of the non-negative super-martingales $\mathscr{M}^{(-k)}$ is at most $1$.
Since the processes $\mathscr{M}^{(-k)}$, $k=1,\dots,n$, have non-increasing trajectories, and since the respective $\mathscr{M}_1^{(-k)}$, are identically distributed, if we pick $D_\gamma$ so that for some $k$ (and so then all $k$)
\begin{equation}
\label{eq:E_process_bound}
E\left[\frac1{D_\gamma}
\max_{\tau\in\mathcal{T}}
T^{(\tau,-k)}(X,Z) \mid\mathcal{F}_0\right] \le 1,
\quad (a.s.),
\end{equation}
where $X\sim P$ (i.e., if $X$ is sampled from the baseline distribution) is independent of the $\sigma$-algebra $\mathcal{F}_0$ and of $Z$,
then the processes $\mathscr{M}^{(-k)}$ (and then $\mathscr{M}$) become super-martingales and then also e-processes.
We abbreviate
\begin{equation}
\label{eq:discrepancies_max}
F^{(-k)}(x,S) = 
\max_{\tau\in\mathcal{T}}
T^{(\tau,-k)}(x,S), 
\quad
x \in \mathbb{R}^p, 
S\in\mathcal{S}_s.
\end{equation}
Note that this function is non-negative and (upper-)bounded by $s$.
This means that~\eqref{eq:E_process_bound} holds if we ensure that $D_\gamma$ is measurable with respect to $\mathcal{F}_0$ and satisfies
\[
E\left[F^{(-k)}(X,Z) \mid\mathcal{F}_0\right] \le D_\gamma,
\quad (a.s.),
\]
where $X$ (independent of $Z$ and $\mathcal{F}_0$) represents {\emph{new}} data from the baseline.

\begin{remark}
The traditional way of calibrating the e-process would be to specify a null model $\mathcal{P}_0$ and set
\[
D_\gamma \ge
\sup_{P\in\mathcal{P}_0}
E_{X\sim P}\left[
F^{(-k)}(X,Z)
\right],
\]
thus not making use of the information in $\mathcal{F}_0$.
Each scaled variable then becomes an e-value.
While this is also available in our setting, it would be at the expense of having to formulate a specific $\mathcal{P}_0$, and be able to compute or bound the supremum~\citep{shafer2011test,vovk2021values}.
We do not pursue this here and instead calibrate the e-process \emph{empirically} and doing away with the need to specify a null model explicitly.
\end{remark}

It would then be natural to set $D_\gamma$ to the conditional expectation
\[
E\left[F^{(-k)}(X,Z)\mid\mathcal{F}_0\right] =
\frac{1}{{p\choose s}}\sum_{S\in\mathcal{S}_s}
\int F^{(-k)}(x,S)\, dP(x),
\] that can be computed for any $k=1,\dots,n$, but it depends on the unknown distribution $P$.
However, since $X$ is independent of $Z$ and of $\mathcal{F}_0$, 
the conditional expectation above can also be written as
\[
E\left[\int F^{(-k)}(x,Z)\,dP(x) \mid\mathcal{F}_0\right],
\]
and since $Z\sim {\rm Unif}(\mathcal{S}_s)$ we have that the above is the expectation of
\begin{equation}
\int F^{(-k)}(x,Z)\,dP(x) \mid\mathcal{F}_0 \sim
{\rm Unif}
\Big(\big\{
\int F^{(-k)}(x,S)\,dP(x): S \in \mathcal{S}_s 
\big\}\Big).
\label{eq:uniformdistribution}
\end{equation}

Next, note that for each $S \in \mathcal{S}_s$, the following average of identically distributed, bounded random variables,
\[
F(S) =
\frac1n \sum_{k=1}^n F^{(-k)}(X_k,S),
\]
is an unbiased estimator for each integral $\int F^{(-k)}(x,S)\,dP(x)$.
We then propose to consider the conditional distribution
$ {\rm Unif}
\Big(\big\{
F(S): S \in \mathcal{S}_s 
\big\}\Big),
$
as a proxy for the conditional distribution in~\eqref{eq:uniformdistribution}.
We set
\[
D_\gamma =
Q^{(\gamma)}[F(Z)\mid \mathcal{F}_0] \vee 1,
\]
for some $0<\gamma<1$, where $Q^{(\gamma)}$ represents the quantile of level $\gamma$ of the distribution of the quantity inside the square brackets.
Note that the $D_\gamma$ are based on sample quantiles of the collection $\big\{F(S): S \in \mathcal{S}_s \big\}$ which can be computed explicitly from the baseline data $X_1,\dots,X_n$.

\subsubsection{Calibrating the norming constant $\gamma$}
\label{section:framework:E_process:norming}
Select any $\delta>0$ 
(eventually depending on $n$ and/or $p$) 
and compute the distribution of $F(Z)\mid\mathcal{F}_0$ from the data.
Then, identify all $\gamma$ such that
\[
Q^{(\gamma)}[F(Z)\mid \mathcal{F}_0] > 
\delta + 
E[F(Z)\mid \mathcal{F}_0],
\]
and set,
\begin{equation}
\hat\gamma =
\hat\gamma_\delta =
\inf_{\gamma >0} \left\{
\gamma:
Q^{(\gamma)}[F(Z)\mid \mathcal{F}_0] > 
\delta + 
E[F(Z)\mid \mathcal{F}_0]
\right\}.
\label{eq:set_gammacalibration}
\end{equation}

We aim to overestimate the conditional expectation of the discrepancy $F^{(-k)}(X,Z)$ by slightly overshooting an estimator thereof.
While one is free to pick $\delta>0$, if $\delta$ is too large, then the set in~\eqref{eq:set_gammacalibration} may be empty.
In such case, even the order statistic $Q^{(1)}[F(Z)\mid \mathcal{F}_0] = \max\big\{F(S): S \in \mathcal{S}_s \big\}$ does not ensure a ``gap'' $\delta$ with the conditional expectation.
On the other hand, setting $\delta$ to be less than the smallest difference between any two different elements in
$\big\{F(S): S \in \mathcal{S}_s \big\}$
is always a valid choice
as the set in~\eqref{eq:set_gammacalibration} is then only empty if the conditional distribution $F(Z)\mid \mathcal{F}_0$ is degenerate.

\subsubsection{Accuracy of the empirical calibration}
\label{section:framework:E_process:theorem}
Note that for any $x,S$, we have $0\le F^{(-k)}(x,S)\le s$ (a.s.), and that by construction $D_{\hat\gamma}\ge1$.
Consider now for $X\sim P$ independent of $\mathcal{F}_0$, the event
\[
A = \big\{
D_{\hat\gamma} < 
E[F^{(-k)}(X,Z)\mid\mathcal{F}_0]
\big\}.
\]
Since $\hat\gamma$ ensures
$
D_{\hat\gamma} > \delta + E[F(Z)\mid \mathcal{F}_0],
$
the event $A$ implies the event
\[
B=\{E[F^{(-k)}(X,Z)\mid\mathcal{F}_0]-E[F(Z)\mid \mathcal{F}_0]>\delta\}.
\]
This means that we can cover the entire sample space with $A^c \cup B$
so that
\[
E\left[
\frac{F^{(-k)}(X,Z)}{D_{\hat\gamma}}\mid\mathcal{F}_0\right] \le 
1 + s\,1_B,
\]
and so, for any $\delta>0$ and any $k=1,\dots,n$, taking expectation,
\[
E\left[
\frac{F^{(-k)}(X,Z)}{D_{\hat\gamma}}\right]
\le
1 + 
s\,P\Big(
E[F^{(-k)}(X,Z)\mid\mathcal{F}_0]-
E[F(Z)\mid \mathcal{F}_0]>\delta
\Big).
\]

This automatically implies, in particular, that the process $\mathscr{M}/(1+s)$ is an e-process, albeit a rather \emph{conservative} one.
Indeed, since $F(Z)$ is an unbiased estimator for each $F^{(-k)}(X,Z)$, we expect that the upper bound is of order
\[
1 + s\, o(1) =
1 + o(1),
\]
for fixed $s$, as $n,p$ go to $\infty$.
This implies that $\mathscr{M}$ itself is (asymptotically) a well calibrated e-process and the (conservative) scaling by $1+s$ is (asymptotically) unnecessary.

The following theorem makes the convergence precise and makes explicit the tradeoffs coming from the choice of $\delta$ and other user parameters such as $s$.

\begin{theorem}
\label{theorem:bound}
Consider observations conforming with the~\eqref{equation:quantiles}-\eqref{equation:observations} and let $P$ represent the distribution of the data under the baseline.
Let $n,p\in\mathbb{N}$ represent 
the number of periods, and 
the number of observations per period,
respectively.

Select a finite number of quantile levels $\mathcal{T}\in(0,1)^{|\mathcal{T}|}$ and suppose that~\footnote{
This is the class of all functions with $\beta$ uniformly bounded derivatives.
The precise definition can be found at the end of Chapter 2.1 of ~\cite{schumaker2007spline}.
} $f^{(\tau)}\in\mathcal{H}_\beta(L)$, $L>0$, for each $\tau\in\mathcal{T}$.
Consider the stochastic process $\mathscr{M}$ defined in Section~\ref{section:framework:E_process} and let $d,m\in\mathbb{N}$ be, respectively, the number of knots and the order of the splines~\eqref{equation:function_model} used to model the conditional quantiles.
Finally, we assume that the joint distribution of the data vectors is absolutely continuous, and that the respective densities are bounded away from $0$ and $\infty$ at the quantiles of interest.

Then, if for some $\delta>0$, the process $\mathscr{M}$ as calibrated as in Section~\ref{section:framework:E_process:norming},
\[
E_P\mathscr{M}_t
\le 
1 + s \min(1,\, \omega), \quad t\in\mathbb{N},
\]
where, 
if $m\ge\beta$, 
for two universal constants $K_1, K_2>0$, 
depending only on $\beta$, and $L$, we have
\[
\omega = 
K_1\,
\left(1+\frac{3s}{\delta}\right) |\mathcal{T}|
\frac{d^{1+2\beta}}{np} 
+
K_2\, 
\frac{s^2}{\delta\, d^{\beta}}
+
2 \exp\left(-\frac{2n\delta^2}{9s^2}\right).
\]
\end{theorem}
\begin{proof}
The proof of this theorem is in Section~\ref{section:proof}.
\end{proof}

The immediate corollary of this theorem is that 
irrespectively of the baseline $P$,
as long as the assumptions of the theorem hold,
\[
P\Big(
\mathscr{M}_t > 1/\alpha
\Big) 
\le
P\Big(
\mathscr{M}_1 > 1/\alpha
\Big) 
\le 
\alpha \Big(1 + s\, \min\big(1,\omega\big)\Big) =
\alpha \Big(1 + s\, \min\big(1,o_P(1)\big)\Big),
\]
where the inflation factor in the upper bound is the price to pay for not having to specify a model for $P$.

The process $\mathscr{M}$ is typically initialized with data coming from the alternative; 
as such, it sits above the threshold and Ville's inequality ensures that it is unlikely for the process to be under the null.
In fact, given its structure, as long as a deviation is detected along any quantile under consideration, the process should take large values.
Reciprocally, once we return to baseline, the value of the e-process is determined by the latest e-value which will then be small.
Once $\mathscr{M}$ drops below the threshold, we have accumulated quite some evidence to support an RtB as no significant deviation is detected along any of the considered quantiles.

\subsubsection{Discussion of the bound}
\label{section:framework:E_process:discussion}

The bound provided by Theorem~\ref{theorem:bound} depends on quite a few quantities so it is worth it to discuss the impact of each of these quantities on the bound.
The upper bound can be used to assess by how much the reference upper bound of 1 is potentially being overshot.
If that is deemed too excessive, then one can further scale the process (eventually by $1+s$) for sharper control.
Appendix~\ref{section:numerics:computational_cost} contains numerical experiments to illustrate the overshoot.
These suggest that the bound is conservative.

It might seem odd that the upper bound is monotone decreasing in $\delta$.
This would suggest taking rather large $\delta$ to minimise the bound but note that the norming constant increases with $\delta$ so that one is then more conservative.
The parameter $d$ controls the flexibility of the splines resulting in a trade-off between two of the terms.
Taking $d$ of order $(np)^{1/(1+3\beta)}$ balances the two terms so that they become of order $(np)^{-\beta/(1+3\beta)}$.
The requirement $m\ge\beta$ cannot be checked in practice since the smoothness of the conditional quantiles is not known.
If $m<\beta$, then the bound is still valid with $\beta$ replaced by $m$.
The last term depends on $n$ but not $p$ which might seem odd given that we have $np$ data points and one might expect all asymptotics to be driven by $np$.
This term comes from an application of the law of large numbers and controls the convergence of an average of $n$ data vectors of a bounded function supported in $\mathbb{R}^p$ to its expectation.
Hence the dependence on $n$ but not $p$.


Only one of the terms in the bound is affected by the number of quantiles under consideration, which we think of as being fixed.
Considering more quantiles does increase the computational cost (see Appendix~\ref{section:numerics:computational_cost}) but enables a finer inspection of the signal which may be necessary if the intervention is expected to have only a subtle effect on the signal.
The range of the discrepancy measure based on which $\mathscr{M}$ is built is determined by $s$ so it is natural for the upper bound to be monotone increasing in $s$.
Large $s$ also has a computational cost that is discussed in Appendix~\ref{section:numerics:computational_cost}.

\section{Numerical results}
\label{section:numerics}


We now illustrate our procedure on synthetic data.
Figure~\ref{fig:Numerics_data_example_00} illustrates a time series sampled under the baseline for one period of 24h, and then under the intervention for another period of 24h, and the underlying conditional data distribution is illustrated in Figure~\ref{fig:Numerics_quantiles_example_00}.
The line in Figure~\ref{fig:Numerics_quantiles_example_00} corresponding $\tau=0.5$, $f^{(0.5)}$ is just a sine wave in $[0,24]$, and the same sine wave with another higher frequency added to it in $[24,48]$.
The remaining $f^{(\tau)}$ are just $f^{(0.5)} + \Phi^{-1}(\tau)$, where $\Phi^{-1}$ is the quantile function of the standard Gaussian distribution.

\begin{figure}[!h]
\centering
\includegraphics[width=0.99\linewidth]{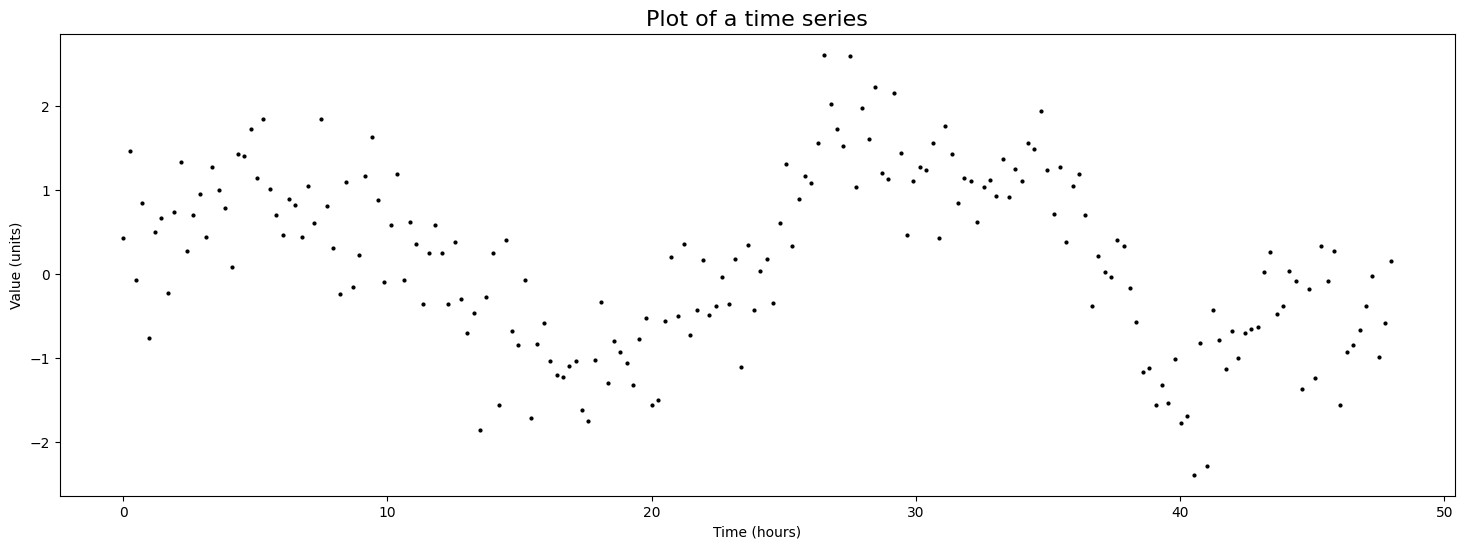}
\caption{
An example of data generated from the model.
There are 200 observations, the first 100 of which were collected under the baseline, followed by 100 under the intervention.
}
\label{fig:Numerics_data_example_00}
\end{figure}


\begin{figure}[!h]
\centering
\includegraphics[width=0.99\linewidth]{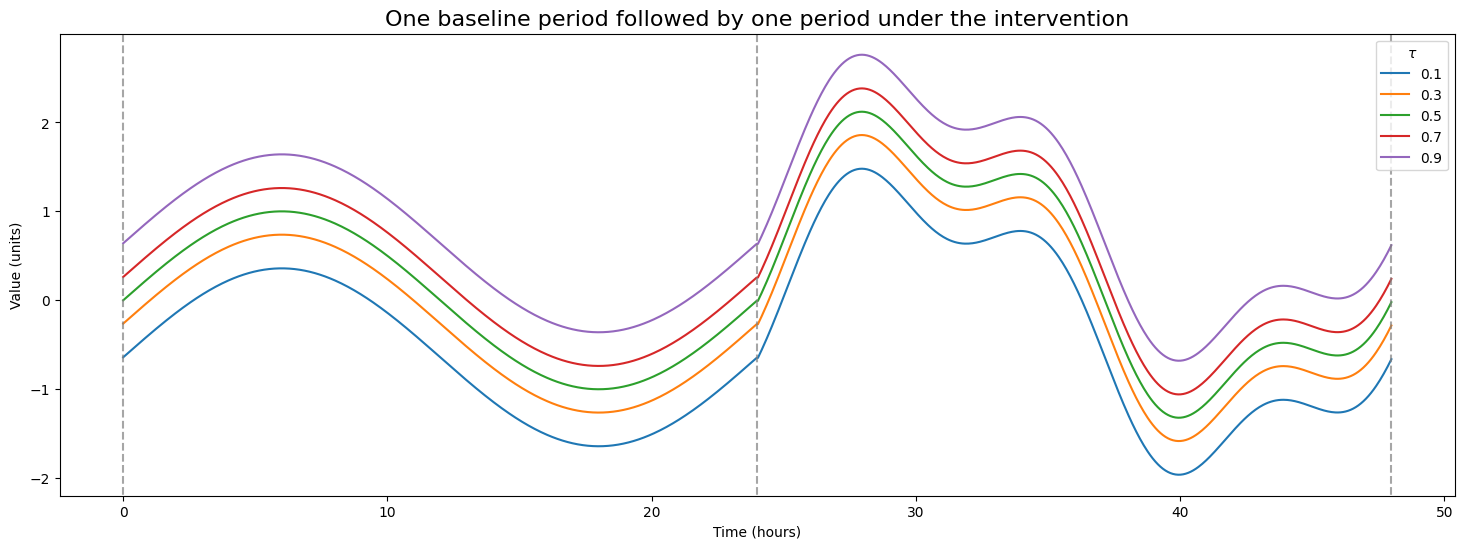}
\caption{
A few conditional quantiles of the baseline and intervention.
Under the intervention, the amplitude of the signal somewhat increases and there are fluctuations at a higher frequency.}
\label{fig:Numerics_quantiles_example_00}
\end{figure}


In Table~\ref{tab:results} we report the results of applying our procedure for different combinations of 
number of baselines ($n$), 
number of observations per period ($p$),
size for the sub-samples ($s$),
and degrees of freedom for the cubic B-spline fits ($\rm{df}$).
In all cases we took 
$\mathcal{T}=\{0.1, 0.2, \cdots, 0.9\}$,
$5$ periods under the intervention ($r=5$), followed by
$5$ baseline periods after the end of the intervention ($m=5$).
For each possible difference between the estimated RtB and the real RtB, we report in a Monte Carlo run of 1000 simulations, the proportion of simulations that attained that difference.
We took $\alpha=0.05$.

\begin{table}[!h]
\centering
\begin{tabular}{c|cccc|lllllll}
\multicolumn{1}{c}{}           & \multicolumn{4}{c}{Parameters} & \multicolumn{7}{c}{Differences}           \\
\multicolumn{1}{c|}{Experiment} & $n$  & $p$   & $s$ & $\rm{df}$ & -5   & -4  & -3  & -2  & -1  & 0    & 1   \\ \hline
1                              & 4    & 4094  & 64  & 16        & 
2.1    & 1.5 & 1.8 & 1.4 & 2.5 & 90.1   & 0.6 \\
2                              & 16   & 1024  & 64  & 16        & 
2.4  & 1.7 & 2.2 & 1.5 & 1.8 & 90.2 & 0.2 \\
3                              & 16   & 256   & 64  & 16        & 
0.5 & 0.3 & 0.7 & 0.9 & 1.2 & 95.6 & 0.8 \\
4                              & 16   & 256   & 8   & 16        & 96.7 & 3.1 & 0.1 & 0.1   & 0   & 0    & 0   \\
5                              & 16   & 258   & 64  & 8         & 
2.4    & 1.4 & 1.7 & 1.5 & 1.6 & 90.8 & 0.6 \\ \hline
\end{tabular}
\caption{Percentage of simulations that each difference between detected RtB and real RtB was obtained.
Since $r=m=5$, this difference can take values in $\{-5,\dots,5\}$; differences not reported in the table were never observed.
In all experiments the threshold of $1/0.05$ is used.}
\label{tab:results}
\end{table}

Experiments 1--3 and 5 produce comparable results, illustrating that $n$ and $p$
can be traded off (note however that they are not fully interchangeable due to the
exponential term in the bound), and that the precision of the calibration does not
hinge on a particular choice of $\mathrm{df}$: since the fitting step precedes the
calibration step, any choice of $\mathrm{df}$ is valid and simply results in different
test statistics.
Experiment 4 is the exception: the small value of $s$ means we may select time
points at which the two signals differ little, causing a loss of power and early
stopping.
A more principled remedy is to pick $s$ closer to $p/2$.\footnote{
There is of course a computational cost associated with large $\binom{p}{s}$; in our
implementation we cap the number of subsets considered at $2^{20}$.}

It is also noteworthy that the procedure is very unlikely to stop too late and miss the RtB:
across the total 5000 experiments, only in 22 cases (0.44\%) did the procedure stop too late and only then only one time unit late.

Since from~\eqref{eq:empirical_distribution_max_discrepancy} the null distribution of the maximal discrepancy is available,
one can directly compute p-values
\[
W_i = \frac{1}{|S_s|}\sum_{S\in S_s} \mathbf{1}\!\left\{F(S) > \tfrac{1}{n}\sum_k F^{(-k)}(X_{n+i},Z)\right\}.
\]
We consider two natural alternatives.
The first converts $W_i$ into an e-value via Shafer's calibrator $f(w)=\kappa w^{\kappa-1}$
($\kappa=0.5$; cf.~\citealp{shafer2011test}) and builds an e-process via a running minimum (Table~\ref{tab:results_comp_e_values}).
The second applies sequential Bonferroni: declare RtB the first time
$\max(W_1,\ldots,W_i)>\alpha/10$, equivalent to $\min(1/W_i)<10/\alpha$ (Table~\ref{tab:results_comp_p_values}).
The same data realisations were used across all three procedures.

\begin{table}[!h]
\centering
\begin{tabular}{c|lllllll}
\multicolumn{1}{c}{}           & \multicolumn{7}{c}{Differences}           \\
\multicolumn{1}{c|}{Experiment} & 
-5   & -4  & -3  & -2  & -1  & 0    & 1   \\ \hline
1                              &  
6.6    & 7.4 & 4.8 & 3.1 & 3.4 & 74.7   & 0.0 \\
2                              & 
7.7  & 5.8 & 5.1 & 4.9 & 3.1 & 73.4 & 0.0 \\
3                              &  
3.0 & 1.6 & 2.1 & 1.8 & 1.9 & 89.6 & 0.0 \\
4                              &  
57.8 & 18.1 & 7.9 & 4.2   & 3.0   & 9.0    & 0.0   \\
5                              &  
6.5    & 5.6 & 4.6 & 3.4 & 3.1 & 76.8 & 0.0 \\ \hline
\end{tabular}
\caption{
Analogue to Table~\ref{tab:results} reporting the results using e-values obtained from calibrated p-values. 
We omit the parameters of each experiment as these are the same as in Table~\ref{tab:results}.
}
\label{tab:results_comp_e_values}
\end{table}

\begin{table}[!h]
\centering
\begin{tabular}{c|llllllll}
\multicolumn{1}{c}{}           & \multicolumn{8}{c}{Differences}           \\
\multicolumn{1}{c|}{Experiment} & 
-5   & -4  & -3  & -2  & -1  & 0    & 1 & 2   \\ \hline
1                              &  
2.9    & 2.1 & 1.9 & 0.8 & 1.0 & 91.0   & 0.3 & 0.0 \\
2                              & 
2.3  & 2.1 & 1.8 & 2.0 & 1.6 & 90.0 & 0.2 & 0.0 \\
3                              &  
0.9 & 0.6 & 0.8 & 0.4 & 0.4 & 96.2 & 0.7 & 0.0 \\
4                              &  
48.8 & 17.5 & 9.6 & 6.3   & 3.5   & 12.7    & 1.5 & 0.1   \\
5                              &  
1.7    & 2.2 & 2.6 & 1.9 & 1.5 & 90.1 & 0.0 & 0.0 \\ \hline
\end{tabular}
\caption{
Analogue to Tables~\ref{tab:results} and~\ref{tab:results_comp_e_values} reporting the results for using p-values sequentially with a Bonferroni correction.
We omit the parameters of each experiment as these are the same as in Table~\ref{tab:results}.
}
\label{tab:results_comp_p_values}
\end{table}

The calibrated e-value approach loses power due to the calibrator (cf.~\citealp{vovk2021values,ramdas2024hypothesis}), as our results confirm.
The Bonferroni approach requires a pre-specified horizon, and as a process its expectation is not bounded by 1 under the null.
Our approach matches Bonferroni while providing anytime-valid guarantees, and outperforms the calibrated e-value approach.
In either of the two approaches the upside should be reiterated, though:
no empirical calibration is necessary to compute the p-values and these can be computed directly without any modeling assumptions.
Nonetheless, our approach matches the Bonferroni approach (which requires a pre-specified horizon) while providing anytime-valid guarantees, and outperforms the calibrated e-value approach (confirming that direct empirical calibration avoids the power loss that calibrator functions introduce).

\section{Real data application}
\label{section:case_study}

In the field of maternal and obstetric care, multiple attempts have been made in identifying a return to baseline from data, when the signal being monitored is the maternal or fetal heart rate (MHR and FHR, respectively), and the intervention is the administration of corticosteroids~\citep{de2008differential, noben2019fetal, bester2023changes, bester2024changes}.

Antenatal corticosteroids are administered to accelerate fetal maturation in
anticipation of preterm birth.
The timing of delivery relative to the RtB is highly relevant: birth before the RtB may
undermine the therapeutic benefit, while delaying unnecessarily exposes mother and
fetus to additional risk.
Prior analyses of this effect (De Heus et al., 2008; Noben et al., 2019; Bester et al.,
2023, 2024) relied on non-parametric tests (Kruskal-Wallis, Friedman) with
post-hoc corrections (Dunn's test, Bonferroni) applied to short summarising
windows of around five minutes, yielding population-level median effects rather
than patient-specific inference.
 
\subsection{The NIEM-O data}
We consider the data from one patient, originally collected for the NIEM-O clinical study, 
in which women with high-risk pregnancies were unobtrusively monitored for several days via eCTG, from which maternal and fetal heart rate can be derived. For additional details, we refer to the study protocol~\citep{de2025continuous}.
Some important details about the study protocol include the primary goal of the study, which was \emph{not} the investigation of the return to baseline, and the fact that the eCTG measurements were collected under free living conditions (i.e.,  women could move freely), implying a number of artefacts such as movement artefacts and signal interruptions (e.g.~taking off the measuring device for walking or taking a shower). 
For the considered patient there are no measurements available before the administration of the corticosteroids (start of the intervention).
The data are recorded at a high-frequency (4 Hz), and are not preprocessed prior to analysis.

After administration (at 18:00 PM) the patient was continuously monitored for 6 days, but the free living conditions lead to interruptions in monitoring.
To mitigate artifacts, and deal with the incomplete nature of the monitoring, we only considered data collected between 23:30 PM and 2:30 AM. This window was present for the 6 days.
We consider the last night of measurements available (\emph{day 6}, the furthest away from the intervention) as baseline.
This follows analogous studies; cf.~\citep{bester2023changes, bester2024changes, noben2019fetal}, and is motivated by the pharmacokinetics of
the specific corticosteroids~\citep{bester2024changes}, and the duration of the drug’s effect
on fetal HR and HRV~\citep{verdurmen2013influence, verdurmen2018influence,  noben2019fetal}.
For the main analysis, we opted for a conservative approach and considered solely the last (and thus furthest) night available as baseline. 
In a secondary stability analysis, we considered the last two days (\emph{day 5} and \emph{day 6}) as baseline.
 
\subsection{Application of the method to real data}
Our approach requires at least two baseline periods (since we hold out one period for comparison, and use remaining data for fitting).
In the main analysis, only the data from \emph{day 6} was taken as baseline (one baseline, thus), and we opted to divide this baseline segment of 150 minutes into either 3, 5, 6, or 10 contiguous sub-segments (of 50, 30, 25, or 15 minutes, respectively.)
While the baseline segment that we selected was fairly stable, one should interpret the fitted quantiles as estimates of the quantiles of the average distribution (over sub-segments); cf.~\cite{Einmahl:2026aa}.
As such, our procedure detects after how many sub-segments (periods of either 50, 30, 25, or 15 minutes) the signal has returned to the baseline:
specifically, we check if, for each sub-segment, that segment is compatible with data distributed according to an average sub-segment from the baseline.

The total number of observations in one period (150 minutes at a sampling frequency of 4 Hz) was 36000.
The RtB framework then involves the choice of a number of parameters.
In our implementation we took 
$\mathcal{T}=\{0.1, 0.2, \dots, 0.9\}$ (quantile levels under consideration).
We settled on $s=100$ and $\mathrm{df}=d+m=30$ (degrees of freedom for the spline fits): lower values of $s$ produced unstable
results, while larger values increased computational cost without appreciable gain;
lower values of $\mathrm{df}$ failed to capture the complexity of the signal, and
no appreciable difference was seen for $\mathrm{df}$ up to $100$.


\subsection{Results}
\subsubsection{Main analysis  ($6^{th}$ night considered as baseline)}
The RtB method can signal a return to baseline of the fetal heart rate (FHR) already after the first night after corticosteroids administration. 
We can see from Figure~\ref{fig:results_multi1blocks} that the significance line (at $20=1/\alpha$) is crossed at the 8th, 6th, 7th, or 10th sub-segments when dividing into 3, 5, 6, or 10 sub-segments, respectively (cf. Table~\ref{tab:summaryresults}.)

\begin{table}[!h]
\centering
\begin{tabular}{cccc}
Sub-segments & 
Duration (min) & 
RtB (sub-segment)& 
RtB (hours)\\
\hline
\hline
3   & 50 & 8 & 59h 50m\\    
5   & 30 & 6 & 35h\\        
6   & 25 & 7 & 35h\\        
10  & 15 & 10 & 13h\\       
\hline
\end{tabular}
\caption{
Summary of results. 
RtB (sub-segment) refers to the number of the sub-segment at which a return to baseline is detected. RtB (hours) refers to the number of hours after first administration at which RtB was detected.
Example: RtB of 8(=3+3+2) with 3 sub-segments indicates that the 2nd segment of the 3rd night an RtB had already taken place.
The times in the final column indicate that that much time after the start of the intervention an RtB had been detected (intervention was 5h 30 min prior to start of monitoring.)
}
\label{tab:summaryresults}
\end{table}

At first glance, the results may not look entirely consistent but this is not the case.
Note that a return to baseline with 50 minute long sub-segments indicates that a whole sub-segment of 50 minutes was deemed statistically indistinguishable (across 9 different conditional quantiles corresponding to 100 randomly picked time points) from what an average sub-segment of 50 minutes under the baseline looks like.
Clearly, this is more stringent than a return to baseline being detected when inspecting the signal with shorter sub-segments.
Note also that an RtB detected in the first segment of a given night means that by that time an RtB had taken place (which does not exclude the possibility that the RtB occurred prior to that segment while no monitoring was taking place). 
Repeating the analysis over sub-segments of different lengths gives the RtB analysis a multi-scale undertone.


\begin{figure}[!ht]
    \centering
    \includegraphics[width=0.80\linewidth]{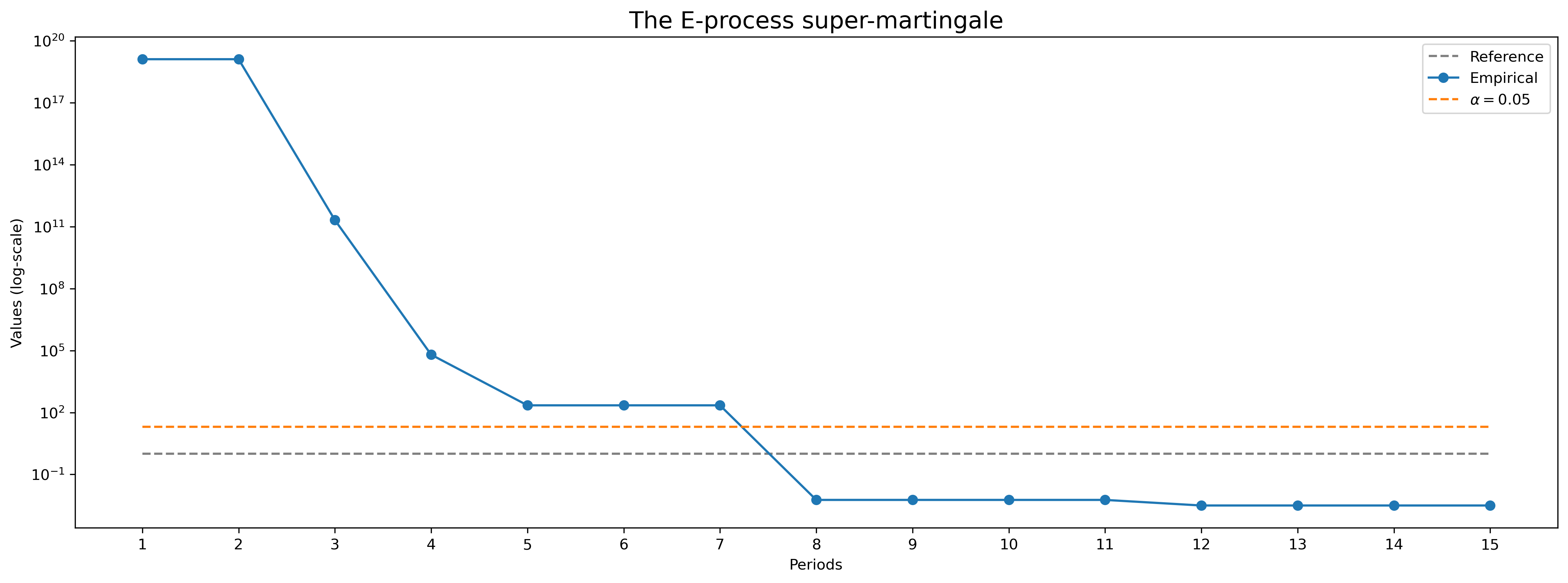}
    \includegraphics[width=0.80\linewidth]{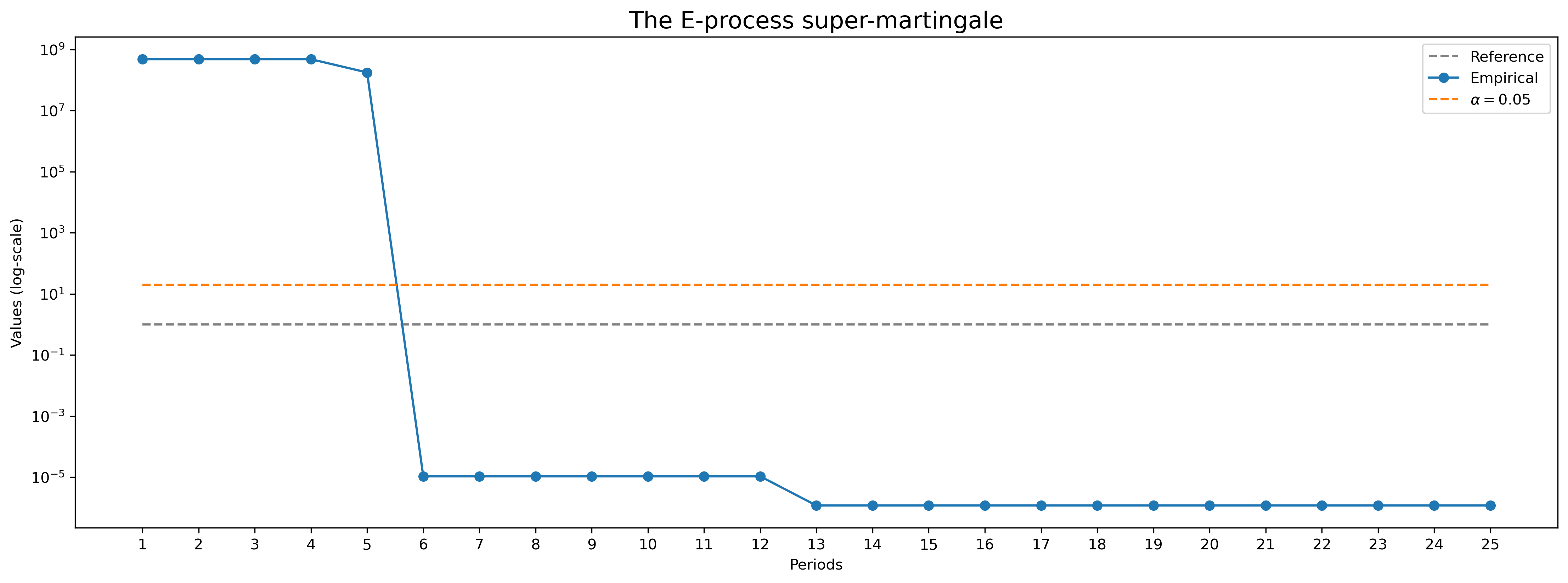}
    \includegraphics[width=0.80\linewidth]{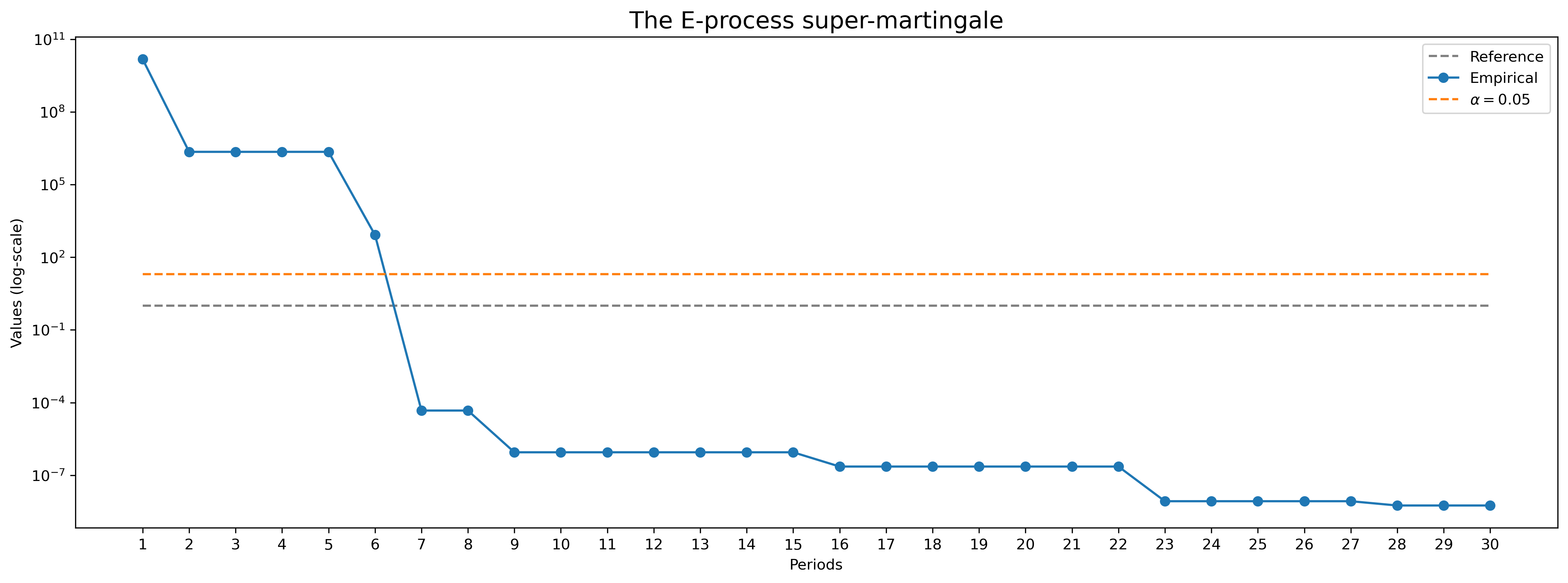}
    \label{fig:results_multi1blocks10}
    \centering
    \includegraphics[width=0.80\linewidth]{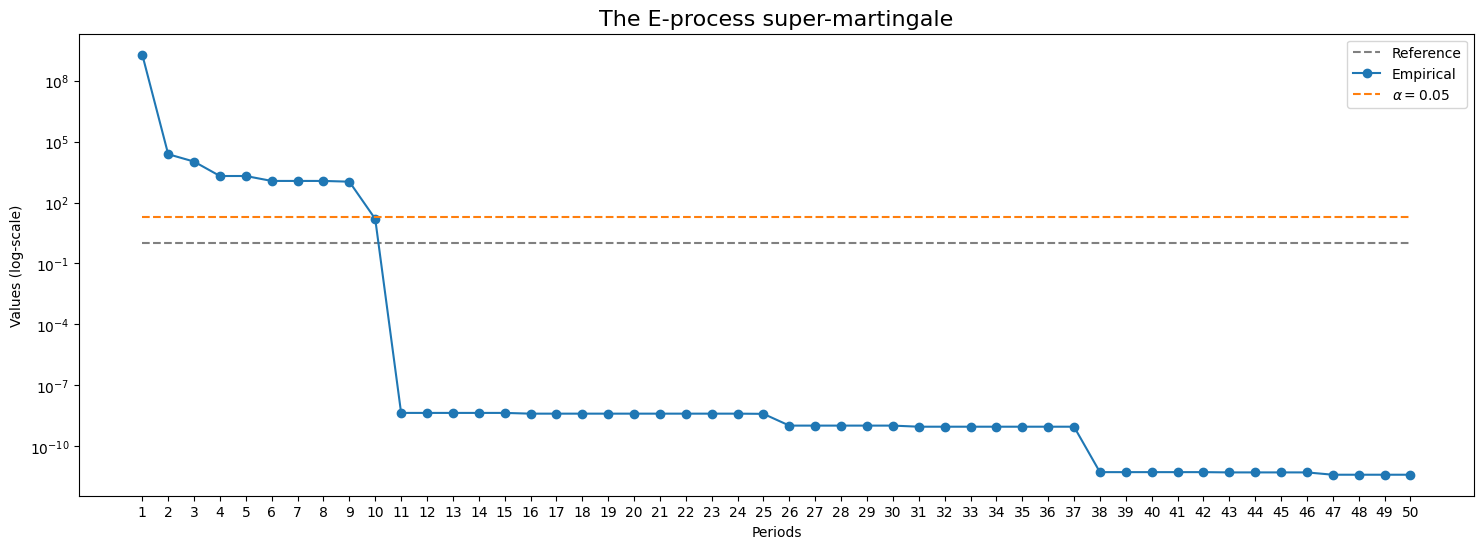}
\caption{
(log-)e-processes corresponding to sequentially inspecting sub-segments of respectively (top to bottom) 50, 30, 25, and 15 minutes.
The threshold of $20 =1/\alpha$ (orange dashed line) is crossed after respectively (top to bottom) 8, 6, 7, and 10 sub-segments.
These correspond to the results in Table~\ref{tab:summaryresults}.
}
\label{fig:results_multi1blocks}
\end{figure}


\subsubsection{Stability analysis ($5^{th}$ and $6^{th}$ night considered as baseline)}
In this case, two baseline segments are available and division in sub-segments is strictly speaking not needed.
However, we repeat the analysis from the previous section now with more baseline data available.
We identify RtB that are in line with the results of the main analysis (cf. Table~\ref{tab:summaryresults_sensitivity}) with a RtB after about 35h.

\begin{table}[!h]
\centering
\begin{tabular}{cccc}
Sub-segments & 
Duration (min) & 
RtB at sub-segment & 
RtB detected after\\
\hline
\hline
3   & 50 & 5 & 35h 50m\\ 
5   & 30 & 6 & 35h\\ 
6   & 25 & 11& 36h 40m\\ 
10  & 15 & 4 & 12h\\ 
\hline
\end{tabular}
\caption{
Summary of results (sensitivity analysis).
Results from repeating the analysis in Table~\ref{tab:summaryresults} with two nights as baseline. 
}
\label{tab:summaryresults_sensitivity}
\end{table}

\subsubsection{Relation to the results of the literature}
Our finding is consistent with the literature, where a return to initial conditions was not found prior to 24 hours~\citep{noben2019fetal}
(although methodological differences being significant: 
ECG-derived signal, 5-minute windows with pre-processing, population-level analysis).
Beyond consistency, the comparison highlights several strengths of the present
approach: finer temporal resolution, subject-specificity
without requiring multi-patient samples, use of multiple quantile levels rather
than location alone, robustness due to use of quantile regression,
and anytime-valid type~I error control that accommodates
any number of tests without multiple testing corrections.




\section{Conclusions}
\label{section:conclusions}

We have introduced a sequential procedure for detecting a return to baseline in
high-frequency monitoring data, where the baseline distribution is unknown, the
intervention effect is unconstrained, and error control must remain valid under
continuous monitoring and optional stopping.
The procedure constructs an empirically calibrated e-process from the baseline
data alone, requiring no parametric null model and remaining inherently
subject-specific.

Our theoretical contribution is Theorem~1, which provides explicit finite-sample
bounds on the calibration error as a function of the amount of baseline data ($n$,
$p$), the smoothness of the conditional quantile functions ($\beta$), the number of
spline knots ($d$), the subsampling size ($s$), and the gap parameter ($\delta$).
These bounds make precise the sense in which the empirically calibrated e-process
converges to a well-calibrated one, and clarify the trade-offs that arise in parameter
selection.
Together with the conservative-but-valid finite-sample guarantee, the procedure
avoids any inflation of the type~I error due to multiple sequential comparisons, since
the e-process framework handles this through Ville's inequality by construction.

The simulations and the clinical case study based on the NIEM-O data confirm that
the procedure behaves as anticipated in realistic settings.
Both the main analysis (single baseline night) and the sensitivity analysis (two
baseline nights) consistently identify a return to baseline of the fetal heart rate
within approximately 35 hours of corticosteroid administration, in agreement with
the pharmacokinetic evidence and with coarser analyses in the existing literature, while remaining patient-specific.
A notable feature of the framework is its multi-scale character: by varying the
sub-segment length, one can probe the signal at different temporal resolutions,
with longer sub-segments imposing a stricter criterion for declaring an RtB.

Quite a few points deserve some comment.
While we assume a quasi-periodic model for the baseline, this is not an intrinsic requirement.
In fact, one can consider completely different types of data and associated models, one can also include covariates, and detection does not necessarily need to be based on comparing conditional quantiles.
The finite-sample bound in Theorem~\ref{theorem:bound} relies on the conditional quantile
functions belonging to the H\"{o}lder class $\mathcal{H}_\beta(L)$, a smoothness condition that cannot be verified in practice; 
however, the calibration step remains valid even if the model is misspecified.
Our numerical experiments suggest the procedure is actually conservative, so that the upper bound can perhaps be improved.
The buffering strategy of Remark~\ref{remark:independence} provides a practical handle on
within-period autocorrelation, but a rigorous treatment of temporal dependence
both within and across periods can also be of interest.

There are also several extensions that may be of practical interest.
Our work focuses on type~I error control but a formal analysis of power is also relevant.
Specifically, one may be interested to understand expected detection delay under the alternative as a function of the strength or other quantities of interest.
Also of interest are regimes where the intervention affects only a narrow range of the distribution, or produces changes in the dependence structure without altering the marginals.
One may also be interested in situations where there is some tolerance in detecting an RtB, or where baseline information from a cohort of patients is available and one would like to borrow strength across individuals while preserving subject-specificity of the inference.
Monitoring several signals simultaneously is straightforward as one can simply produce an e-process out of discrepancies like~\eqref{eq:discrepancies_max} with an extra maximum over the different signals.
A large value for such process indicates a deviation across at least one signal, while a small value for the process indicates no deviation across any signal.

%


\section*{Data Availability}
The clinical data from the NIEM-O study used in Section~\ref{section:case_study} contains personally
identifiable health information and cannot be made publicly available.
All numerical results in Sections~\ref{section:numerics}, ~\ref{section:case_study} and Appendix~\ref{section:numerics:computational_cost} can be reproduced using the code
available at \url{https://github.com/PauloJASerra/RtB_simulations},
which also includes code and seeds used to generate the synthetic datasets used in Section~\ref{section:numerics}.

\section*{Acknowledgments}

All content and code is of the sole authorship of the authors. AI was only used to provide feedback on the manuscript prior to submission, perform spellchecking, and to check for compliance with journal guidelines.
The authors would like to thank Ivar de Vries and Nadine de Klerk for the  helpful discussions and feedback on clinical aspects of the work.

\bibliographystyle{abbrvnat}
\bibliography{reference}

\begin{appendices}

\section{Proof of Theorem~\ref{theorem:bound}}
\label{section:proof}

\noindent\textbf{Decomposition}
Define for 
$x \in \mathbb{R}^p, S\in\mathcal{S}_s$,
\begin{align*}
N^{(\tau)}(x, S) &=
\big|\{j\in S: 
x_j > f^{(\tau)}(t_j)\}\big|,\\
T^{(\tau)}(x,S) &= 
\big|N^{(\tau)}(x, S)-(1-\tau)|S|\big|,\\
F(x,S) &= 
\max_{\tau\in\mathcal{T}}T^{(\tau)}(x,S).
\end{align*}
These parallel the quantities 
$N^{(\tau,-k)}(x, S)$, 
$T^{(\tau,-k)}(x, S)$, 
$F^{(-k)}(x, S)$, from before but with the estimators
$\hat{f}^{(\tau,-k)}$ replaced with the true conditional quantiles
$f^{(\tau)}$.

By simply adding and subtracting some terms, we can rewrite
\[
E[F^{(-k)}(X,Z)\mid\mathcal{F}_0]-
E[F(Z)\mid \mathcal{F}_0] =
H_1-H_2+H_3
\]
where we have, more explicitly,
\begin{align*}
H_1 &= E[(F^{(-k)}-F)(X,Z)\mid\mathcal{F}_0] =
\frac1{{p \choose s}}\sum_{S\in\mathcal{S}_s}
\int (F^{(-k)}(x,S)-F(x,S))\,dP(x)
,\\
H_2 &= \frac1n\sum_{k=1}^nE[(F^{(-k)}-F)(X_k,Z)\mid \mathcal{F}_0] =
\frac1{n{p \choose s}}\sum_{S\in\mathcal{S}_s}\sum_{k=1}^n
\Big(F^{(-k)}(X_k,S)-F(X_k,S)\Big)
,\\
H_3 &= E[\int F(x,Z)\,d(P-P_n)(x)\mid \mathcal{F}_0] =
\frac1{{p \choose s}}\sum_{S\in\mathcal{S}_s} 
\int F(x,S)\,d(P-P_n)(x),
\end{align*}
where $P_n$ represents here the empirical measure of the sample $X_1,\dots,X_n$.

It remains to bound
\[
P\Big(
H_1-H_2+H_3>\delta
\Big) \le
\sum_{i=1}^3
P\Big(|H_i|>\delta/3\Big).
\]

\noindent\textbf{Discrepancy control}
Before we bound these quantities note that by using the triangle inequality and the definition of the functions $F^{(-k)}(x,S)$, we have
\begin{align*}
F^{(-k)}(x,S) &\le
\max_{\tau\in\mathcal{T}}
|T^{(\tau,-k)}(x,S) - T^{(\tau)}(x,S)| +
\max_{\tau\in\mathcal{T}}T^{(\tau)}(x,S)\\
&\le
\max_{\tau\in\mathcal{T}}
|N^{(\tau,-k)}(x,S) - N^{(\tau)}(x,S)| +
F(x,S),
\end{align*}
where we use the reverse triangle inequality in the last step.
The bound also holds if we swap $F^{(-k)}$ and $F$ so we conclude that
\[
|F^{(-k)}(x,S)-F(x,S)| \le
\max_{\tau\in\mathcal{T}}
|N^{(\tau,-k)}(x,S) - N^{(\tau)}(x,S)|.
\]

Consider now,
for $\Delta>0$ (allowed to depend on $n,p$),
the event
\[
G^{(-k)} = 
G_{\Delta}^{(-k)} =
\Big\{
\max_{\tau\in\mathcal{T}}
\sup_{t\in[0,1]}
|\hat f^{(\tau,-k)}(t) - f^{(\tau)}(t)|
\le \Delta
\Big\},
\]
and the random indicator functions on $\mathbb{R}^p$,
\[
c^{(-k)}(x,S) = 
c_\Delta^{(-k)}(x,S) =
\prod_{j\in S} 1
\Big\{
|(f^{(\tau)}(t_j)+\hat f^{(\tau,-k)}(t_j))/2 -x_j|> \Delta/2
\Big\}.
\]
Under the event $G^{(-k)}$, if $c^{(-k)}(X,S)=1$ for some independently sampled $X\sim P$ and some $S\in\mathcal{S}_s$, then this implies that, 
for every $\tau\in\mathcal{T}$, 
each coordinate of $X$ with index in $S$ lies 
either 
above both $f^{(\tau)}$ and $\hat{f}^{(\tau,-k)}$,
or
below both $f^{(\tau)}$ and $\hat{f}^{(\tau,-k)}$.
As a consequence, 
for every $\tau\in\mathcal{T}$, every $S\in\mathcal{S}_s$, and each $k$,
\[
N^{(\tau,-k)}(X,S) = N^{(\tau)}(X,S),\]
which implies
$|F^{(-k)}(x,S)-F(x,S)|=0$.

\noindent\textbf{Bound for the $H_1$ term}
Note that $|F^{(-k)}(x,S)-F(x,S)|$ can be decomposed into
\[
|F^{(-k)}(x,S)-F(x,S)| c^{(-k)}(x,S) + 
|F^{(-k)}(x,S)-F(x,S)| \big(1-c^{(-k)}(x,S)\big),
\]
but, under $G^{(-k)}$, $|F^{(-k)}(x,S)-F(x,S)|=0$ whenever $c^{(-k)}(x,S)=1$
so we conclude that under $G^{(-k)}$,
\[
|F^{(-k)}(x,S)-F(x,S)| \le
s\, 1\{c^{(-k)}(x,S)=0\},
\]
so that under $G^{(-k)}$, 
\[
|H_1| \le
\frac1{{p \choose s}}\sum_{S\in\mathcal{S}_s}
\int s\, 1\{c^{(-k)}(x,S)=0\}\,dP(x) =
s\, P\big(c^{(-k)}(X,Z)=0\mid\mathcal{F}_0\big).
\]
We conclude that for each $k$,
\[
P(|H_1|>\delta/3) \le
P(G^{(-k),c}) + 
P\big(G^{(-k)}, 
s P(c^{(-k)}(X,Z)=0\mid\mathcal{F}_0)>
\delta/3
\big),
\]
where we use $P(A,B)$ to denote $P(A\cap B)$.

We have that 
$P(c^{(-k)}(X,Z)=0\mid\mathcal{F}_0) = 
1 - P(c^{(-k)}(X,Z)=1\mid\mathcal{F}_0)$, 
and that under $G^{(-k)}$
\begin{align*}
P(c^{(-k)}(X,Z)=1\mid \sigma(Z),\mathcal{F}_0) &= 
(1-O_P(\Delta))^s =
\exp(-O_P(s\Delta)) \\ &=
1-O_P(s\Delta),
\end{align*}
as long as $\Delta$ can be chosen so that 
$\Delta=o_{n,p}(1)$.
Under $G^{(-k)}$, this can be made more explicit by imposing an upper bound on the P-probability of an observation falling within a $\Delta$-neighborhood of one of the quantiles under consideration.
In any case, we conclude after a use of Markov's inequality that,
\[
P(|H_1|>\delta/3) \le 
P(G^{(-k),c}) +
O_{n,p}\left(\frac{s^2\Delta}{\delta}\right).
\]
For now we leave the bound as is since $P(G^c)$ will feature in other terms.

\noindent\textbf{Bound for the $H_2$ term}
This term is similar to the previous one but with $P$ replaced with $P_n$ so it has to be handled slightly differently.
We apply Markov's inequality followed by the tower property to get that
\begin{align*}
P(|H_2|>\delta/3) &\le
\frac{3}\delta
\frac1{n{p \choose s}}\sum_{S\in\mathcal{S}_s}
\sum_{k=1}^n E|F^{(-k)}(X_k,S)-F(X_k,S)|\\ & =
\frac{3}\delta
\frac1{{p \choose s}}\sum_{S\in\mathcal{S}_s}
E\int|F^{(-k)}(x,S)-F(x,S)|\,dP(x),
\end{align*}
for any $k=1,\dots,n$, and $X\sim P$ independent of $\mathcal{F}_0$, since the differences $|F^{(-k)}(X_k,S)-F(X_k,S)|$ are identically distributed across $k$ for each fixed $S$.

We introduce the indicator of $G^{(-k)}$ and proceed as with $H_1$ to bound the integral with the help of the indicator function $c^{(-k)}$,
\[
\frac{3s}\delta P(G^{(-k),c}) +
\frac{3s}\delta
E1_{G^{(-k)}}\,P\big(
c^{(-k)}(X,Z)=0\mid\mathcal{F}_0\big).
\]
We can now use the law of total expectation to condition on $Z$ and use to bound from before to conclude that
\[
P(|H_2|>\delta/3) \le 
\frac{3s}\delta P(G^{(-k),c}) +
O_{n,p}\left(\frac{s^2\Delta}{\delta}\right).
\]

\noindent\textbf{Bound for the $H_3$ term}
We can write
\[
H_3 =
\int g(x)\,d(P-P_n)(x),
\quad\text{for}\quad
g(x) =
\frac1{{p \choose s}}\sum_{S\in\mathcal{S}_s}F(x,S),
\]
where $g$ is a bounded function.
Since $|g(x)|\le s$, and the $X_1,\dots,X_n$ for a random sample, we can use Hoeffding's inequality to get
\[
P\big(|H_3|>\delta/3\big) \le 
2 \exp\left(-\frac{2n\delta^2}{9s^2}\right).
\]

\noindent\textbf{Finalizing the bound}
Putting all together, we have
\[
E\left[
\frac{F^{(-k)}(X,Z)}{D_{\hat\gamma}}\right]
\le
1 + 
s\left(1+\frac{3s}{\delta}\right)
P(G^{(-k),c}) +
O_{n,p}\left(\frac{s^3\Delta}{\delta}\right) +
2s 
\exp\left(-\frac{2n\delta^2}{9s^2}\right).
\]
It remains to bound, for any $k$, the probability
\begin{align*}
P(G^{(-k),c}) &=
P\big(\max_{\tau\in\mathcal{T}}
\sup_{t\in[0,1]}
|\hat f^{(\tau,-k)}(t) - f^{(\tau)}(t)|
> \Delta\big) \\ &\le
\sum_{\tau\in\mathcal{T}}
P\big(
\sup_{t\in[0,1]}
|\hat f^{(\tau,-k)}(t) - f^{(\tau)}(t)|
> \Delta\big)
\end{align*}

Next note that for any 
$d,m,n\in\mathbb{N}$,
$k=1,\dots,n$,
$\rho>0$,
$\tau\in(0,1)$, and
$t\in[0,1]$ (w.l.g.),
we have
\[
\hat f^{(\tau,-k)}(t) =
f_{\hat\theta^{(\tau,-k)}}(t) =
\sum_{i=1}^{d+m}
\hat\theta_i^{(\tau,-k)} B_i^m(\rho\, t).
\]
If $f^{(\tau)}\in\mathcal{H}_\beta(L)$ is the true underlying $\tau$-quantile function, then there exists $\theta^{(\tau)}$ so that
\[
\sup_{t\in[0,1]}
|f_{\theta^{(\tau)}}(t) - f^{(\tau)}(t)|^2 \le
\phi_{m,\beta, L}(d),
\]
where, as long as $m\ge\beta$, we can take 
\[
\phi_{m,\beta,L}(d) =
C_\beta\, 
L^2\, 
d^{-2\beta},
\]
for some universal constant $C_\beta$ depending only on $\beta$; cf.~Corollary 6.26 of~\cite{schumaker2007spline}.
We conclude that
\[
|\hat f^{(\tau,-k)}(t) - f^{(\tau)}(t)|^2 \le
2|f_{\hat\theta^{(\tau,-k)}}(t) - 
f_{\theta^{(\tau)}}(t)|^2 +
2\phi_{m,\beta, L}(d),
\]
since $(a+b)^2\le2(a^2+b^2),a,b\in\mathbb{R}$, and then, since the B-spline basis function are non-negative and partition unity, the difference in the upper bound is
\[
\Big|
\sum_{i=1}^{d+m}
(\hat\theta_i^{(\tau,-k)}-\theta_i^{(\tau)}) B_i^m(\rho\, t)
\Big|^2 \le
\sum_{i=1}^{d+m}
(\hat\theta_i^{(\tau,-k)}-\theta_i^{(\tau)})^2,
\]
by Jensen's inequality, so that the upper bound no longer depends on $t$.

We conclude that as long as 
$\Delta^2 > 2\phi_{m,\beta, L}(d)$,
then
\begin{align*}
P(G^{(-k),c}) \le
\sum_{\tau\in\mathcal{T}}
P\left(
\sum_{i=1}^{d+m}
(\hat\theta_i^{(\tau,-k)}-\theta_i^{(\tau)})^2
> 
\frac12\Delta^2 -
\phi_{m,\beta, L}(d)
\right),
\end{align*}
is under control.
In fact, this upper bound vanishes for appropriately vanishing $\Delta$ since by the CLT for regression quantiles, their squared error is
\[
O_P\left(\frac{d+m}{(n-1)p}\right),
\]
which holds under the assumption made in the theorem.
Markov's inequality then gives us
\[
P(G^{(-k),c}) \le 
O_{n,p}\left(
\frac{|\mathcal{T}|(d+m)}{(n-1)p(\Delta^2-2\phi_{m,\beta,L}(d))}
\right).
\]

\noindent\textbf{Minimizing the bound}
The second and third terms in the upper bound depend on $\Delta$ which can be taken freely as long as $\Delta^2 > 2\phi_{m,\beta, L}(d)$.
To balance these two terms we then take for instance $\Delta^2 = 3\phi_{m,\beta, L}(d)$ if we pick $d$ large enough so that $\phi_{m,\beta, L}(d)$ is greater than the risk delivered by the CLT.
With some extra trivial bounding, this leads to the bound in the statement of the theorem.

\section{Discussion on computational cost}\label{section:numerics:computational_cost}

The calibration procedure outlined in Section~\ref{section:framework:E_process:norming} relies on computing the distribution of $F(Z)\mid\mathcal{F}_0$.
This only has to be done once, after all baseline data has been collected.

Computing this distribution involves recomputing $F$ from the data for each subset $S$ of $\{1,\dots,p\}$ of cardinality $|Z|=s$;
evaluating $F$ itself implies carrying out a hold-out procedure over each of the $n$ observations (to compute each of the $F^{(-k)}$ that are averaged to get $F$).
The computational cost of computing this distribution is then $n{p\choose s}$ times the computational cost of evaluating each $F^{(-k)}(S)$ for a given $k,S$.

Evaluating each $F^{(-k)}(S)$, in turn, involves fitting $|\mathcal{T}|$ splines of order $m$ with $d+m$ parameters to the respective conditional quantile curves, each based on $(n-1)p$ data points.
In our implementation we used the \emph{statsmodels} library in Python; cf.~\cite{statsmodels}.
While we could not find any explicit information about the computational cost of this particular implementation, we expect it to scale linearly with the number of quantiles, quadratically with the number of observations, and linearly with the number of parameters, leading to a computational cost of order $n^2p^2(d+m)|\mathcal{T}|$.
Finally, after the model is fitted, we evaluate each observation which should incur in a cost of $d+m$.

We conclude that the calibration procedure has a computational cost of order no more than
\[
n^3{p\choose s}p^2(d+m)^2|\mathcal{T}|.
\]
We ran a few simulations to get a better feeling of how the computational cost actually scales with each parameter in practice.

We took as a baseline a dataset
\footnote{The specific data distribution is not relevant when assessing the computation time but we took the same signal as in Section~\ref{section:numerics}.}
with 
$n=4$ periods under the baseline,
$p=2^5=32$ observations per period,
$m=3$ for the order of the B-splines,
$\mathcal{T}=\{0.1, 0.5, 0.9\}$ three equally-spaced quantile levels between $0.1$ and $0.9$, and finally,
$df = 12 = d+m$ B-spline coefficients which mean working with $d=12-m$ knots in total.

We ran the calibration procedure 10 times and report the average execution time which on a laptop was
$0.36$ 
seconds. 
We then took each of the parameters above one by one and reran the calibration for multiples of the initial value keeping all other parameters fixed.
In Figure~\ref{fig:computation_time_n} we show the effect of increasing $n$.

\begin{figure}[!h]
\centering
\includegraphics[width=0.75\linewidth]{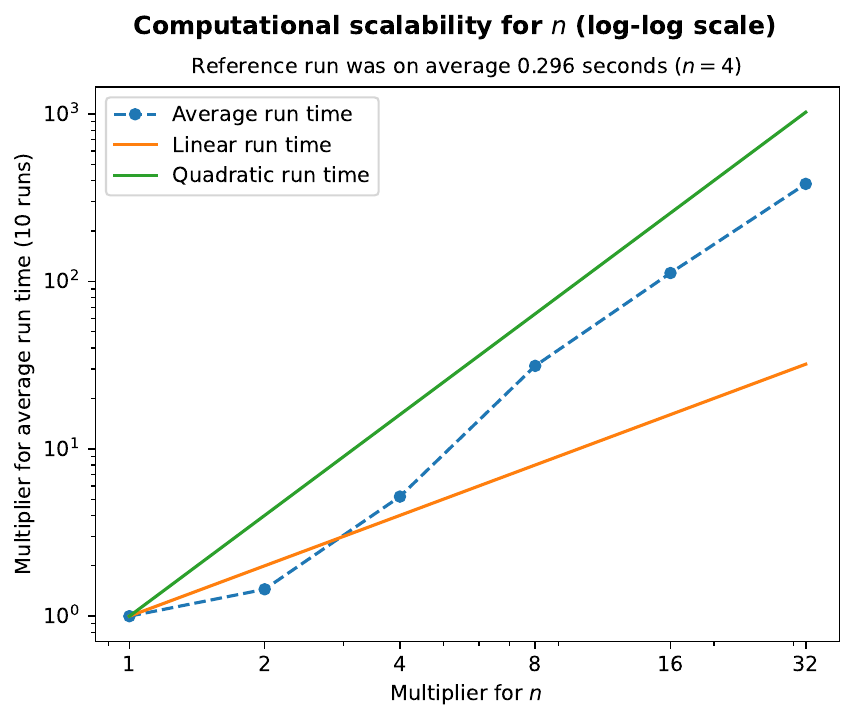}
\caption{Scaling in $n$ of the computational time.}
\label{fig:computation_time_n}
\end{figure}

We plot on a log-log scale the multiplier for $n$ and the multiplier for the corresponding execution time.
So for instance the point at roughly $(16,100)$ means that running the code with $n=4\times 16$ increases the execution time (roughly) 100-fold from the reference $0.36$ seconds.
We also report for reference what linear- and quadratic growth would look like.
Comparing slopes, the figure suggests that the computational cost seems to actually align better with quadratic (rather than the expected cubic) growth.
This may be related to how the \emph{numpy} library handles vectorized computations.
The plots in Figure~\ref{fig:computation_time_others} present analogous results for the remaining parameters.

\begin{figure}[!h]
\centering
\includegraphics[width=0.49\linewidth]{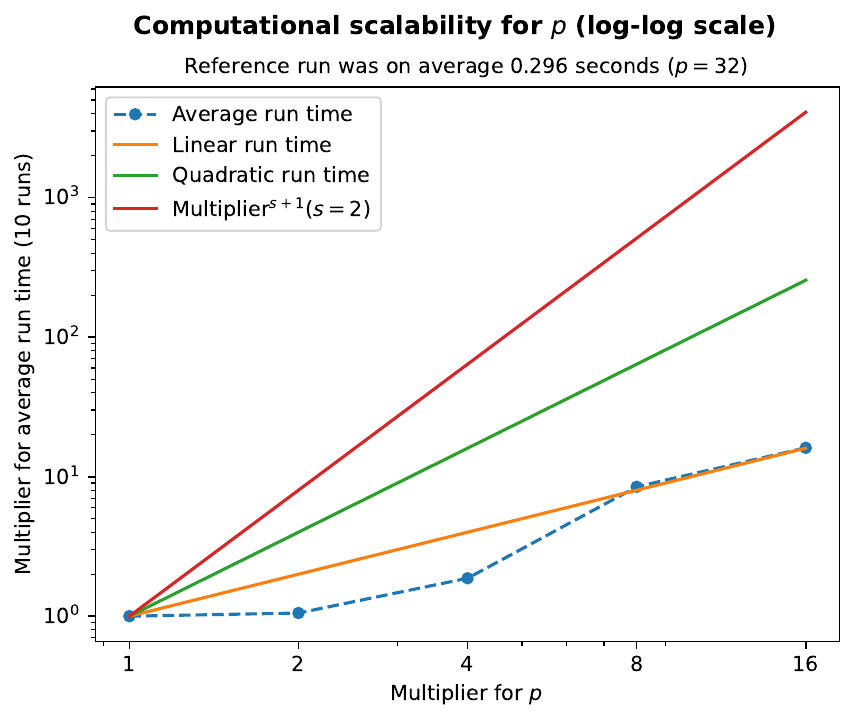}
\includegraphics[width=0.49\linewidth]{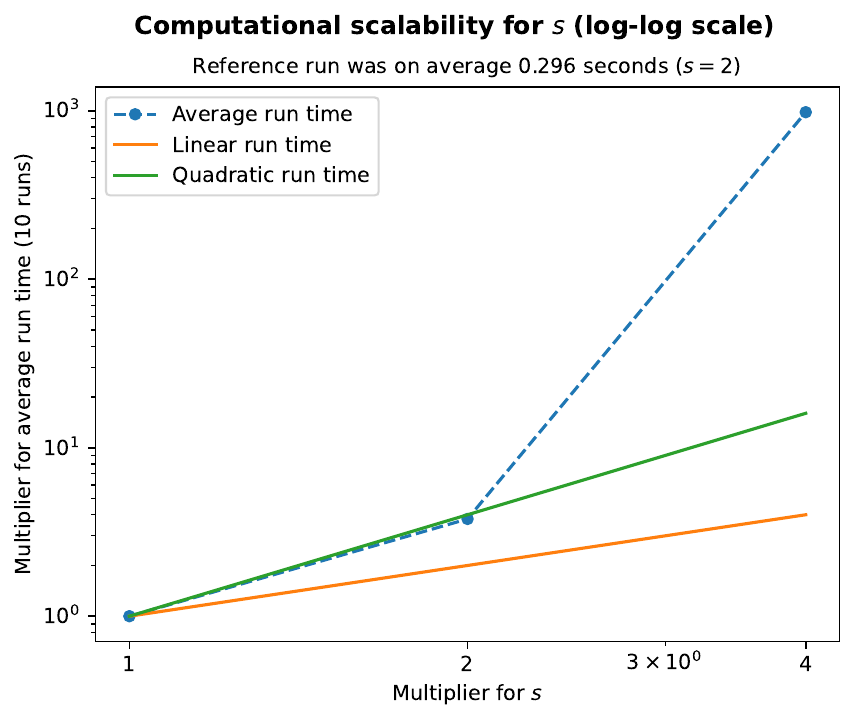}\\
\includegraphics[width=0.49\linewidth]{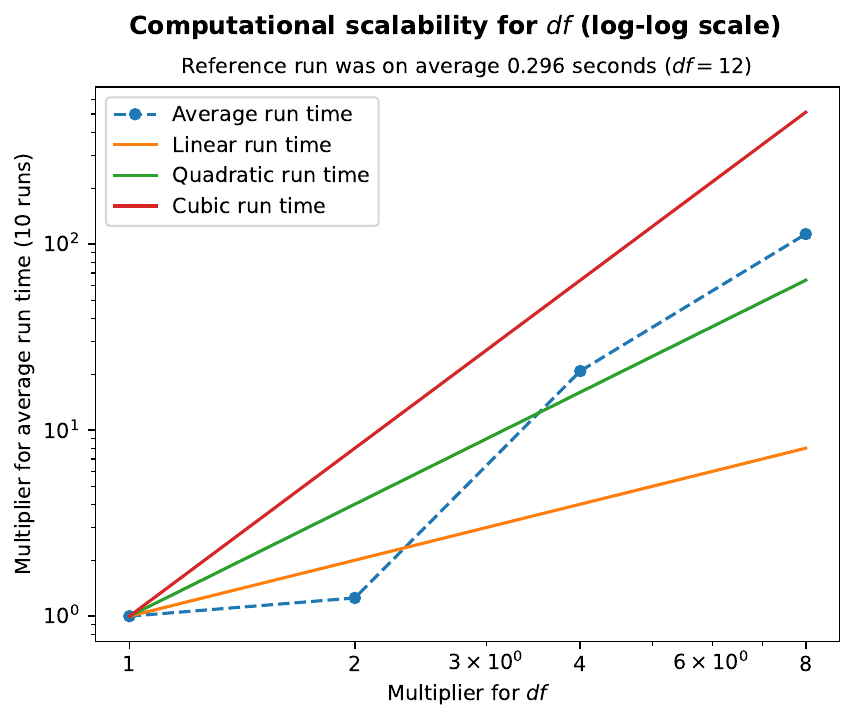}
\includegraphics[width=0.49\linewidth]{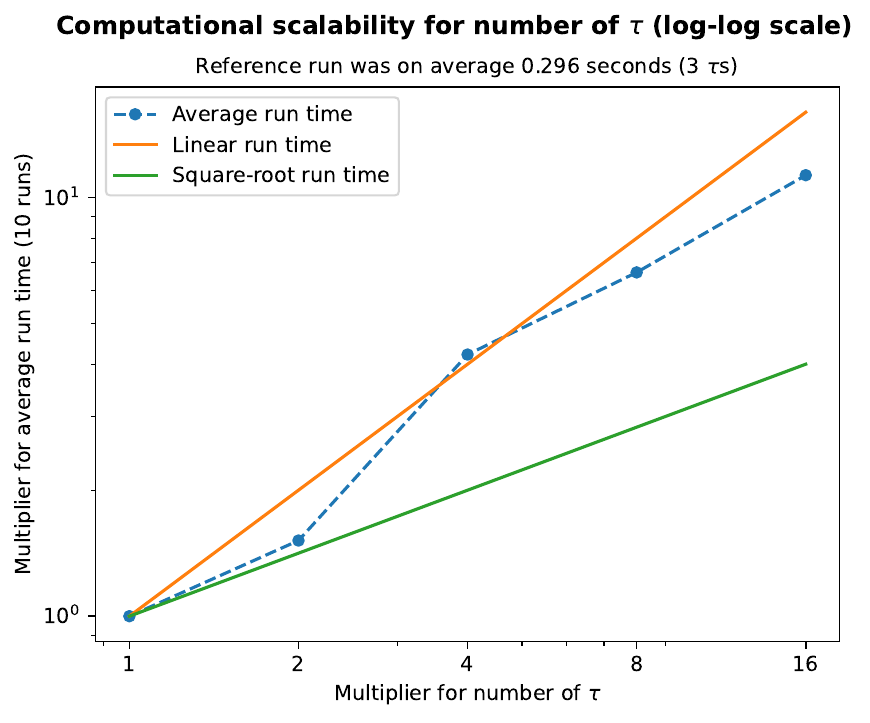}
\caption{Scaling in $p$, $s$, $df$, and $|\mathcal{T}|$ (left to right and top to bottom) of the computational time.}
\label{fig:computation_time_others}
\end{figure}

The scaling for $p$ is as expected aligned with quadratic- or cubic ($s+1=3$) growth.
The scaling in $s$ is by far the least favorable.
Indeed, increasing $s$ has a large impact on the number of sets $S$ that have to be considered.
If the computational cost becomes too steep, one may consider capping the maximal number of subsets that is considered and thus produce a Monte-Carlo approximation of the norming constant.
We implemented this possibility in our code but that cap was not reached in our simulation.
The scaling for the number of coefficients is somewhat steeper than the expected (quadratic) growth.
This may be due to costs associated with storing higher dimensional representations for the fitted line.
Finally, scaling up the number of $\tau$ incurs in a lower than expected cost, likely due to how the several quantiles are fitted simultaneously.

Two final remarks about computation time.
Given the high computational cost of having large $p$, particularly if $s$ is large, one may reduce $p$ at the expense of increasing $n$ via thinning;
for instance, each period is split into two, one containing the odd indexed observations and one the even indexed ones.
In effect we then have $2n$ periods, each with $p/2$ observations.
Increasing $p$ at the expense of reducing $n$ is also possible by merging periods; 
for instance, merging the data from each odd indexed period with the following period leads to $n/2$ periods, each with $2p$ observations.
This may be advantageous if $s$ is relatively small.

A second remark is that the same way that one may impose a cap on the number of subsets $S$ that we consider, one can likewise not hold out every single one of the $n$ baseline periods but instead only a subset of them.
In our numerical experiments we in fact saw that as $np$ increases, the estimated $\hat{f}^{(-k)}$ are quite similar leading to nearly identical discrepancies $F^{(-k)}$ so that averaging over the held out periods, while useful to establish Theorem~\ref{theorem:bound}, may not be so essential in practice.

\subsection{Calibration accuracy}
\label{section:numerics:calibration_acuracy}

The goal of the calibration procedure is to ensure that the random variable distributed like $F(X,Z)$ is scaled so as to have expectation bounded by 1.
In this numerical experiment, we perform a Monte Carlo procedure to numerically approximate the expectation of $F(X,Z)$ scaled by a range o quantiles of $F(Z)$.

Specifically, we took as baseline the baseline signal from Section~\ref{section:numerics} and 
a) sampled data from the baseline (with $n=4$, $p=32$);
b) sampled $Z=z$ (with $s=2$); 
c) computed the distribution of $F(Z)$ from the data;
d) used different statistics of this distribution to scale $F(X,z)$ for independent $X$ from the baseline.
The above was repeated $10^4$ times to obtain i.i.d.\ samples of $F(X,z)$ for different norming constants.
Figure~\ref{fig:normalisation_MC} displays the average and (approximate) 95\% error bounds for the scaled $F(X,z)$ and the results for the same experiment with a larger value of $n$.
\begin{figure}[!h]
\centering
\includegraphics[width=1\linewidth]{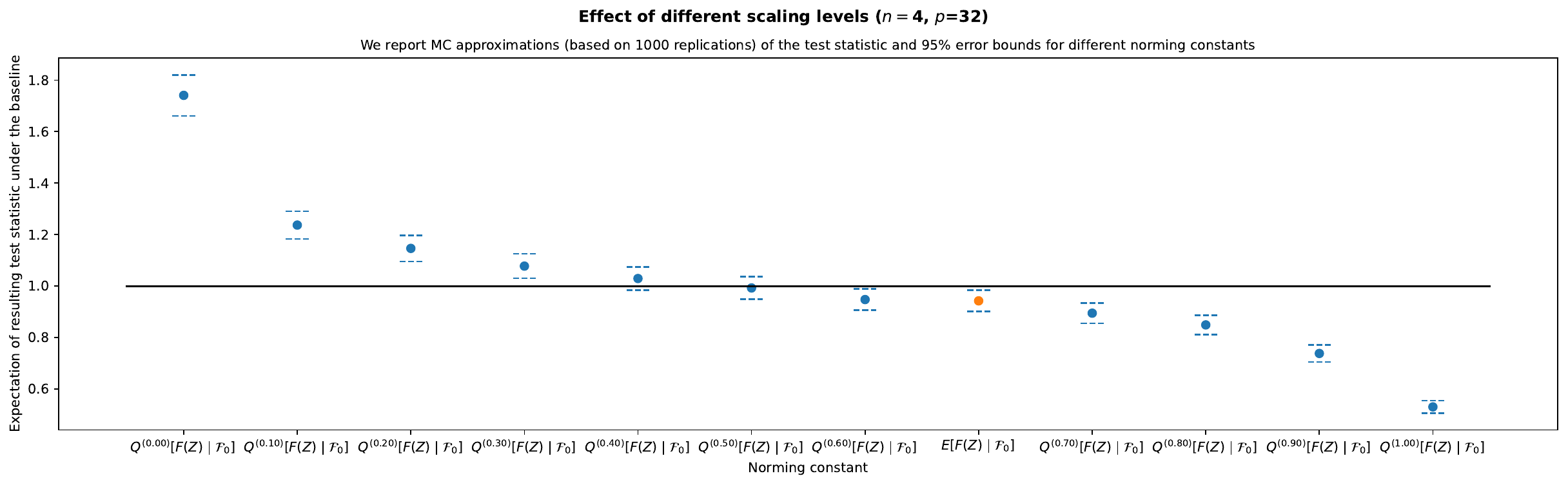}
\includegraphics[width=1\linewidth]{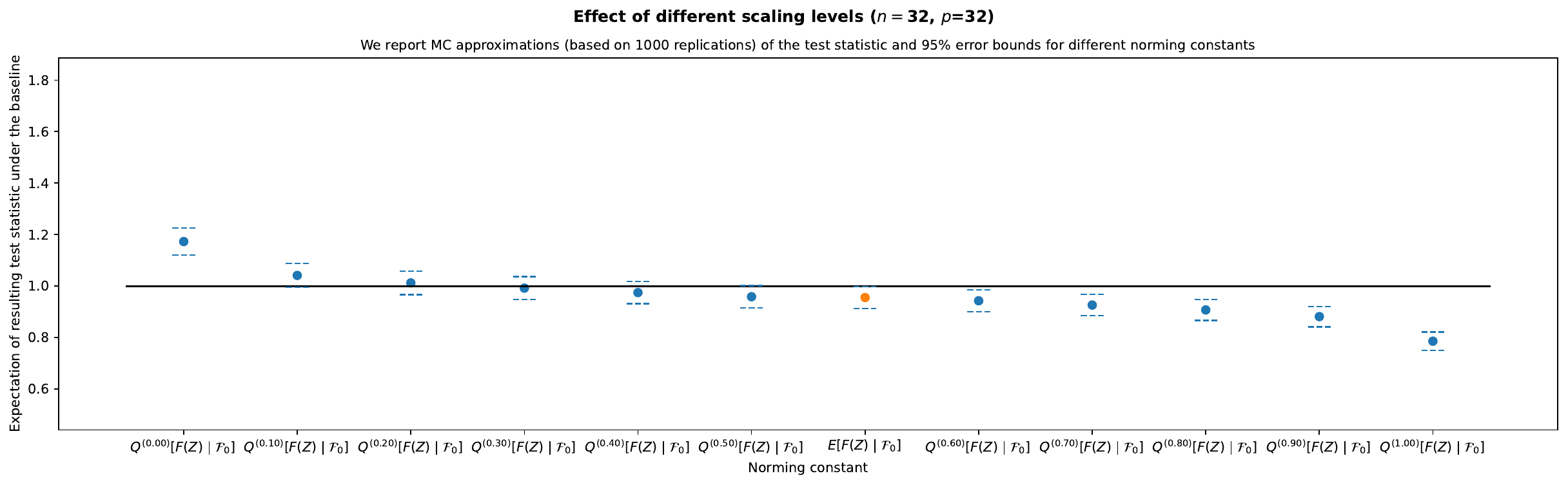}
\caption{Effect of different scaling factors on the discrepancy measure. Rather than reporting different values of $\delta$ (since the results would be dependent on the choice of $s$) we report quantiles of $F(Z)$ and, in orange, the expectation (corresponds to roughly $\delta=0$.) Above we have $n=4$ while below $n=32$.}
\label{fig:normalisation_MC}
\end{figure}

Theorem~\ref{theorem:bound} prescribes that 
(under certain assumptions), scaling by any quantile that sits above the expectation, 
the expectation is asymptotically bounded by 1.
On the horizontal axis of Figure~\ref{fig:normalisation_MC} we report the different statistics of the distribution of $F(Z)$ used to scale $F(X,z)$.

The figure illustrates that while the result of the theorem is rather conservative, the expectations of the test statistic under the baseline is close to 1, even if statistics below the conditional expectation of $F(Z)$ are used.
This is an indication that the entire conditional distribution of $F(Z)$ given the data concentrates (as $n$ grows) around the expectation of the discrepancy $F(X,Z)$.

The fact that the scaling prescribed by Theorem~\ref{theorem:bound} might be somewhat conservative is not a concern.
If an e-process is calibrated with respect to a give null model $\mathcal{P}_0$, presumably there are measures in $\mathcal{P}_0$ for which the expectation of the process is much less than 1.

\end{appendices}

\end{document}